\documentclass[a4paper]{article}
\usepackage{a4wide}
\usepackage{graphicx,xcolor}
\graphicspath{{./images/}}
\usepackage{caption}
\usepackage{subcaption}
\usepackage{empheq}
\usepackage{amsthm}
\usepackage{amssymb, latexsym, mathrsfs}
\usepackage{amsmath}
\usepackage{dsfont}
\usepackage{enumerate}
\usepackage{algorithm}
\usepackage{color}
\usepackage{transparent}
\usepackage{url}
\usepackage[affil-it]{authblk}
\usepackage{booktabs}

\usepackage[utf8]{inputenc}
\usepackage[T1]{fontenc}
\usepackage{lmodern}

\usepackage[american,cuteinductors,smartlabels]{circuitikz}
\usepackage{tikz}
\usepackage{schemabloc}
\usetikzlibrary{shapes,arrows}

\theoremstyle{plain}
\newtheorem{theorem}{Theorem}
\newtheorem{ass}{Assumption}

\newtheorem{rmk}{Remark}
\newtheorem{prp}{Proposition}
\newtheorem{lem}{Lemma}
\newtheorem{definition}{Definition}

\newcommand{\Qset}{\mathbb{Q}}
\newcommand{\Rset}{\mathbb{R}}

\newcommand{\mbf}[1]{\mathbf{#1}}                  

\newcommand{\subss}[2]{{#1}_{[#2]}}

\newcommand{\matr}[1]{
\begin{bmatrix}
    #1
\end{bmatrix}
}

\begin{document}
	\title{\LARGE \bf A consensus-based secondary control layer for stable current sharing and voltage balancing in DC microgrids}
          \author[1]{Michele Tucci%
       \thanks{Electronic address:
         \texttt{michele.tucci02@universitadipavia.it}}}
        \author[2]{Lexuan Meng%
       \thanks{Electronic address: \texttt{menglexuan@gmail.com}; Corresponding author}} 
\author[2]{Josep M. Guerrero%
       \thanks{Electronic address: \texttt{joz@et.aau.dk}} }
\author[3]{Giancarlo Ferrari-Trecate%
       \thanks{Electronic address: \texttt{giancarlo.ferraritrecate@epfl.ch}} }
 \affil[1]{Dipartimento di Ingegneria Industriale e
       dell'Informazione\\Universit\`a degli Studi di Pavia, Italy}
     \affil[2]{Institute of Energy Technology, Aalborg University, Denmark}
	\affil[3]{Automatic Control Laboratory, \'Ecole Polytechnique F\'ed\'erale de Lausanne (EPFL), Switzerland}
     \date{\textbf{Technical Report}\\ March, 2016}

     \maketitle

  \begin{abstract}
  	In this paper, we propose a secondary consensus-based control
  	layer for current sharing and voltage balancing in DC microGrids
  	(mGs). To this purpose, we assume that Distributed Generation Units (DGUs) are equipped with decentralized primary controllers guaranteeing voltage stability. This goal can be achieved using, for instance, Plug-and-Play (PnP) regulators.
  	We analyze the behavior of the
  	closed-loop mG by approximating local primary control loops with
  	either unitary gains or first-order transfer functions. Besides
  	proving exponential stability, current sharing, and voltage balancing, we describe how to design secondary controllers in
  	a PnP fashion when DGUs are added or
  	removed. Theoretical results are complemented by simulations, using a 7-DGUs mG implemented in Simulink/PLECS, and experiments on a 3-DGUs mG.
  	
  \end{abstract}
  
\newpage

\section{Introduction}
\label{sec:intro}
Power generation and distribution are rapidly changing due to the
increasing diffusion of renewable energy sources, advances in energy
storage, and active participation of consumers to the energy market
\cite{ipakchi2009grid, hu2016greener,justo2013ac}. This shift of paradigm has motivated the
development of migroGrids (mGs), commonly recognized as small-scale
power systems integrating Distributed Generation Units (DGUs), storage
devices and loads. 
Since AC power generation is the standard for commercial, residential,
and industrial utilization, several studies focused on the control of AC mGs
\cite{guerrero2013advanced,riverso2014plug,bolognani2013distributed, 7097655, 7500071}. However,
nowadays, DC energy systems are gaining interest
\cite{dragicevic2015dc,elsayed2015dc} because of the
increasing number of DC loads, the availability of efficient
converters, and the need of interfacing DC energy sources and
batteries with minimal power losses \cite{justo2013ac,elsayed2015dc, fairley2012dc}. As reviewed in \cite{elsayed2015dc}, DC mGs are
becoming more and more popular in several application domains, such as
avionics, automotive, marine and residential systems \cite{elsayed2015dc}. 

The basic
problems in control of DC mGs are  voltage stabilization
\cite{dragicevic2015dc,tucci2015decentralized,zhao2015distributed,de2016power,cezar2015stability,hamzeh2016power, zonetti2014globally} and
current sharing (or, equivalently, load sharing), the latter meaning
that DGUs must compensate constant load currents proportionally to given parameters (for example, the converter ratings) and independently of the mG topology and line
impedances. Current sharing is crucial for preserving the safety
of the system, as unregulated currents may overload generators and eventually lead to failures or system
blackout \cite{han2016review}. An
additional desirable goal is voltage balancing, i.e. to keep the average output
voltage of DGUs close to a prescribed level. 
Indeed, load devices are designed to be supplied by a nominal
reference voltage: it is therefore
important to ensure that the voltages at the load buses are spread around this value.

These objectives can be realized using hierarchical control structures. 
In particular, current sharing is often realized by coupling secondary-layer consensus algorithms with primary voltage controllers. Secondary regulators, however, can have a detrimental effect and closed-loop stability of the mG must be carefully analyzed. In order to guarantee stability, several existing approaches focus on specific mG topologies or provide centralized control design algorithms, i.e. the synthesis of a local controller for a DGU requires knowledge about all other DGUs and lines. For instance,  in \cite{behjati2014modular,moayedi2015team} only mG with a bus topology are considered and local control design is performed in a centralized fashion. The synthesis procedures proposed in \cite{6816073,8026170} suffer from similar limitations. Indeed, in these works, the computation of local controllers guaranteeing collective stability relies on the knowledge of a global closed-loop transfer function matrix, and the largest eigenvalue of the Laplacian matrix associated with the communication network.
	Consensus-based secondary controllers for mG with general topologies have been presented in \cite{shafiee2014distributed_b}; also in this case, however, stability and the closed-loop mG equipped with primary and secondary control layers is studied through centralized analysis (i.e. the root locus), or via simulations \cite{dragicevic2015dc}. Similar drawbacks affect the approach in \cite{meng2015modeling}, whereas, in \cite{Setiawan2017}, current sharing is achieved by means of a centralized controller which receives informations from all the DGUs in the mG.
All the synthesis algorithms mentioned above become prohibitive for large mGs. Moreover, they are unsuitable for mGs with flexible structure because, to preserve voltage stability, the plugging-in
or -out of DGUs might require to update all local controllers in
the mG. This motivated the development of \emph{scalable} design procedures
for local controllers as in \cite{tucci2015decentralized, tucci2016improved,zhao2015distributed}.
In \cite{tucci2015decentralized} and \cite{tucci2016improved}, the aim is to stabilize the  voltage only via primary decentralized controllers. These regulators, termed Plug-and-Play (PnP) have the following features: (\textit{i}) the existence of a local controller for a DGU can be tested on local hardware, and control design is cast into an optimization problem, (\textit{ii}) each optimization problem exploits information about the DGU only \cite{tucci2016improved} or, at most, the power lines connected to it \cite{tucci2015decentralized}, and (\textit{iii}) when a DGU is plugged-in, no other DGUs \cite{tucci2016improved}, or at most neighboring DGUs \cite{tucci2015decentralized}, must update their local controllers.

\paragraph{Paper contributions.}
In this paper, we present a secondary regulation scheme for achieving stable current sharing and voltage
	balancing in DC mGs. We assume that the proposed higher-level scheme is build on top of a primary stabilizing voltage control layer; moreover, similarly to
	\cite{zhao2015distributed}, at the secondary level we exploit consensus filters requiring DGUs to communicate in real-time over a connected network. 
	There is, however, a key difference between our approach and the one presented in \cite{zhao2015distributed}. In the latter work, the authors assume DGUs to be controllable \textit{current} sources; on the contrary, we consider primary-controlled DGUs behaving as \textit{voltage} generators. In this setting, by properly choosing the reference values for the output voltages of each DGU, one can always regulate the currents that flow through the power lines, even when load conditions change. In this way, unwanted circulating currents (which may open the safety breakers in order to prevent damages on the devices) can be avoided. Instead, if all the DGUs in the mG are ideal current sources, the primary regulation scheme alone cannot guarantee such control on the flowing currents \cite{zhao2015distributed}. 
	
	At the modeling level, we propose two abstractions for the DGUs controlled with primary voltage regulators: unit gain and first-order transfer function approximation. The first one is used only for tutorial purposes and for developing basic mathematical tools that will allow us to extend the key results to the second (more realistic) approximation of the primary loops.
	
	Another contribution of this paper is the study of the eigenstructure of the product of three matrices ($\mathbb{L} D \mathbb{M}$), where (\textit{i}) $\mathbb{L}$ and $\mathbb{M}$ are the graph Laplacians associated with the electrical
	and the communication graphs, respectively, and (\textit{ii}) $D$ is a diagonal positive definite matrix defining the desired ratios between balanced currents. While several studies focused on the properties of the product of stochastic
	matrices (see e.g. \cite{jadbabaie2003coordination}), which are
	central in discrete-time consensus, to our knowledge weighted products of
	Laplacians received much less attention. In particular, we show that, under two different conditions, $\mathbb{L}D\mathbb{M}$ preserves some key features of Laplacian matrices. In this case, the asymptotic achievement of current sharing and voltage balancing in a globally exponentially stable fashion is proved.
	
	Finally, we provide an experimental validation of our approach using a lab-size DC mG. The results show the robustness of the proposed controllers to non-idealities that are unavoidably present in a real mG.

\paragraph{Paper organization.}	
The paper is organized as
follows. Section \ref{sec:PnP_primary} summarizes the electrical model
of DGUs and PnP controllers. The secondary control layer is developed
and analyzed in Sections \ref{sec:consensus} and
\ref{sec:modeling_and_analysis}. In particular, Section \ref{sec:pnp_sec_ctrl} shows that, similarly to the regulators in \cite{tucci2015decentralized} and \cite{tucci2016improved}, secondary controllers can be designed in a PnP fashion. Section \ref{sec:simulations} demonstrates current
sharing and voltage balancing through simulations in Simulink/PLECS
\cite{allmeling2013plecs}, where non-idealities of real converters
and lines have been taken into account. Finally, in Section \ref{sec:experiments} we present experimental tests performed on a real DC mG.

A preliminary version of the paper will be presented at the 20th IFAC World Congress. Different from the conference version, the present paper includes (\textit{i}) the proofs of Propositions 2-5 and Theorem 1, (\textit{ii}) the more realistic case where primary control loops are approximated with first-order transfer functions, and (\textit{iii}) experimental results and more detailed simulations.

\paragraph{Notation and basic definitions.} 
The cardinality of the finite set $S$ will be denoted with $|S|$. A
weighted directed graph (\emph{digraph}) $\mathcal G = (\mathcal V, \mathcal E, W)$ is
defined by the set of nodes $\mathcal{V}=\{1,\dots, n\}$, the set of edges
$\mathcal E \subseteq  \mathcal V \times \mathcal V$ and the diagonal
matrix $W\in\mathbb{R}^{|\mathcal{E}|\times|\mathcal{E}|}$ with $W_{ii}
=w_i$, where $w_i\in\mathbb{R}$ is the weight associated with the edge $e_i\in\mathcal{E}$. The set of neighbors of node $i\in\mathcal{V}$ is $\mathcal{N}_i
=\{j:(i,j)\in\mathcal{E}\text{ or } (j,i) \in
\mathcal{E}\}$. A digraph $\mathcal{G}$ is \textit{weakly connected} if its undirected
version is connected \cite{lns-v.85}. $Q(\mathcal{G})\in\mathbb{R}^{|\mathcal{V}|\times|\mathcal{E}|}$ is the incidence matrix of $\mathcal{G}$
\cite{grone1990laplacian}. The Laplacian matrix of $\mathcal{G}$ is $\mathcal
L(\mathcal{G}) = Q(\mathcal{G})WQ(\mathcal{G})^T$, and it is independent of the orientation of edges. 

The average of a vector $v\in\mathbb{R}^n$ is $\langle v\rangle=\frac{1}{n}\sum_{i=1}^n v_i$. We denote with $H^1$ the subspace composed by all vectors with zero average \cite{bensoussan2005difference,1643380} i.e. $H^1 = \{v\in\mathbb{R}^n:\langle v\rangle = 0\}$. The space orthogonal to $H^1$ is $H_{\perp}^1$. It holds $H_{\perp}^1 =\{\alpha\mathbf{1}_n,\text{ }\alpha\in\mathbb{R}\}$ and dim$(H_{\perp}^1)=1$. Moreover, the decomposition $\mathbb{R}^n = H^1\oplus H_{\perp}^1$ is direct, i.e. each vector $v\in\mathbb{R}^n$ can always be written in a unique way as
\begin{equation}
\label{eq:decomposition}
v = \hat v + \bar v \hspace{3mm}\text{with }\hat v\in H^1\text{ and }\bar v\in H^1_{\perp}.
\end{equation}
Consider the matrix $A\in\mathbb{R}^{n\times n}$. With
$A(H^1|H^1)$ we indicate the linear map $A:H^1\rightarrow H^1$
(i.e. the restriction of the map
$A:\mathbb{R}^n\rightarrow\mathbb{R}^n$ to the subspace $H^1$). For a
subspace $\mathcal{V}\subset\mathbb{R}^n$, we denote with
$P_{\mathcal{V}}(v)$ the projection of $v\in\mathbb{R}^n$ on
$\mathcal{V}$. The subspace $\mathcal{V}\subset\mathbb{R}^n$ is said to be
$A$\textit{-invariant} if $v\in\mathcal{V}\Rightarrow A v\in\mathcal{V}$. Moreover, with $\lambda_i(A)$, $i=1, \dots, n$, we denote the eigenvalues of $A$.

Let $A\in\mathbb{R}^{n\times n}$ be a matrix with real eigenvalues. The \textit{inertia} of $A$ is the triple $i(A) = (i_{+}(A), i_{-}(A), i_{0}(A))$, where $i_{+}(A)$ is the number of positive eigenvalues of $A$,  $i_{-}(A)$ is the number of negative eigenvalues of $A$, and  $i_{0}(A)$ is the number of zero eigenvalues of $A$, all counted with their algebraic multiplicity \cite{hong1991jordan}. We use $A>0$ for indicating that the real symmetric matrix $A$ is positive-definite.

Laplacian matrices have the key properties summarized in the next
Proposition \cite{agaev2005spectra,godsil2001algebraic,bensoussan2005difference}.
\begin{prp}
	\label{pr:laplacian_prop}
	For a weakly connected graph $\mathcal{G}$ with weights $w_i>0$, $A = \mathcal{L}(\mathcal{G})\in\mathbb{R}^{n\times n}$ has the following properties:
	\begin{enumerate}[(i)]
		\item \label{lapl_1}it has non positive off-diagonal elements;
		\item \label{lapl_2} $\lambda_1(A)\geq\dots\geq\lambda_{n-1}(A)\geq 0 = \lambda_n$;
		\item \label{lapl_3} $\mathrm{Ker}(A)= H_{\perp}^1$ and $\mathrm{Range}(A)= H^1$;
		\item \label{lapl_4} $A(H^1|H^1)$ is invertible.
	\end{enumerate}
\end{prp}
\begin{proof}
	Points (\ref{lapl_1})-(\ref{lapl_3}) are shown, e.g. in \cite{agaev2005spectra, godsil2001algebraic}. Point (\ref{lapl_4}) has been shown in \cite{bensoussan2005difference} with the framework of partial difference equations. Next, we provide a proof based on linear algebra only. We start noticing that the linear map $A(H^1|H^1)$ invertible if it is surjective and injective \cite{lang1987linear}. First, we show the surjectivity of $A$ on $H^1$. By construction, rank$(A) = n-1$ because $$\text{rank}(A) = \mathrm{dim}(\mathrm{Range}(A))=\mathrm{dim}(\Rset^n)-\mathrm{dim}(H_{\perp}^1)=n-1.$$ Moreover, since $A$ is symmetric, each column of $A$ has zero sum and hence is a vector in $H^1$. Since $\mathrm{Range}(A)$ is the column span, then $\mathrm{Range}(A)\subseteq H^1$. Since dim$(H^1)=n-1$, one obtains $\mathrm{Range}(A)=H^1$. This proves that the map $A(H^1|H^1)$ is surjective. \\
	Next, we prove that $A(H^1|H^1)$ is also injective. By definition, this holds if
	\begin{equation*}
	\forall b\in H^1\hspace{3mm}\forall x,y \in H^1\hspace{3mm}(Ax=b \mbox{ and } Ay=b)\Rightarrow x=y.
	\end{equation*}
	Now, $Ax =Ay = b$ implies that $A(x-y) = 0$. It means that $x-y\in \mathrm{Ker}(A)$, therefore $\exists\alpha\in{\Rset}$ such that $x-y = \alpha \mbf{1}_n$. However, since $x-y\in H^1$, $x-y = \alpha \mbf{1}_n$ is verified only for $\alpha = 0$; this leads to $x = y$.  
\end{proof}

\section{Plug-and-play primary voltage control}
\label{sec:PnP_primary}

\subsection{DGU electrical model}
\label{sec:mG_model}
As in \cite{tucci2015decentralized}, we consider a DC mG composed of
$N$ DGUs, whose electrical scheme is shown
in Figure \ref{fig:ctrl_part}. In each DGU, the generic renewable
resource is modeled as a battery and a Buck converter is
used to supply a local load connected to the Point of Common Coupling
(PCC) through an $RLC$ filter. 
Furthermore, we assume
that loads $I_{Li}$ are unknown and treated as current disturbances
\cite{tucci2015decentralized,tucci2016improved}. The controlled variable is the voltage at
each PCC, whereas the control input is the command to the Buck (indicated with $V_{ti}$). From Figure \ref{fig:ctrl_part}, by applying Kirchhoff's
voltage and current laws and exploiting Quasi Stationary Line (QSL)
approximation of power lines \cite{tucci2015decentralized,tucci2016improved}, we obtain
the following model of DGU $i$\footnote{For the detailed model
	derivation, we defer the reader to \cite{tucci2015decentralized}.}

\begin{equation*}
\label{eq:newDGU}
\text{DGU}~i:\hspace{-4.5mm}\quad\left\lbrace
\begin{aligned}
\frac{dV_{i}}{dt} &= \frac{1}{C_{ti}}I_{ti}+\sum\limits_{j\in\mathcal{N}_i}\left(\frac{V_j}{C_{ti} R_{ij}}-\frac{V_i}{C_{ti}R_{ij}}\right)-\frac{1}{C_{ti}}I_{Li}\\
\frac{dI_{ti}}{dt} &= -\frac{1}{L_{ti}}V_{i}-\frac{R_{ti}}{L_{ti}}I_{ti}+\frac{1}{L_{ti}}V_{ti}\\
\end{aligned}
\right.
\end{equation*}
where $(V_{ti},I_{Li})$ are the inputs, $(V_{i},I_{ti})$ the states (i.e. the measured voltage at the $i$-th PCC and $i$-th output current, respectively), $V_{j}$ is the voltage at the PCC of each neighboring DGU $j\in\mathcal{N}_i$, the constants $R_{ti},C_{ti},L_{ti}$ identify the electrical parameter of the $i$-th Buck filter, and $R_{ij}$ is the conductance of the power line connecting DGUs $i$ and $j$ (see Figure \ref{fig:ctrl_part}).

\subsection{Plug-and-play regulators}
\label{sec:pnp_design}
In this Section, we briefly summarize the PnP scalable approach in
\cite{tucci2015decentralized,tucci2016improved} for designing primary decentralized controllers guaranteeing
voltage stability in DC mGs. This will allow us to justify the approximations of primary control loops used in Section \ref{sec:modeling_and_analysis}. Moreover, we will describe local updates that must be performed when DGUs are added or removed. These operations will be mirrored by those described in Section \ref{sec:pnp_sec_ctrl} for updating local secondary controllers, hence showing that both the primary and secondary control layer can be designed in a modular and scalable fashion. 

The local regulator of DGU $i$
exploits measurements of $V_i$ and $I_{ti}$ to compute the command
$V_{ti}$ of the $i$-th Buck converter and make $V_{i}$ track a
reference signal $V_{ref,i}$ (see the scheme in Figure
\ref{fig:ctrl_part}). Each controller is composed of a vector
gain $K_i$ and an integral action is present for offset-free voltage
tracking. The decentralized computation of these vector gains is the
core of PnP controller synthesis: (\textit{i}) the design of $K_{i}$ requires
knowledge of the dynamics of DGU $i$ only \cite{tucci2016improved} or, at most, the parameters of power lines connecting it to its neighbors \cite{tucci2015decentralized}, (\textit{ii}) $K_i$ is automatically
obtained by solving a Linear Matrix Inequality (LMI) problem.

For modeling the \textit{electrical} interactions between multiple DGUs, we represent the mG
with a digraph $\mathcal{G}_{el}=(\mathcal{V}, \mathcal{E}_{el}, W)$ (see the example in Figure \ref{fig:5areasplug}), where (\textit{i}) each node is a DGU with
local PnP controller and local current load, (\textit{ii}) edges $(i,j)$ are
power lines whose orientation define a reference direction for
positive currents, (\textit{iii}) weights are line conductances\footnote{Line inductances $L_{ij}$ are neglected as we assume QSL approximations \cite{tucci2015decentralized, tucci2016improved}, which are reasonable in low-voltage mGs.} $\frac{1}{R_{ij}}$, and
(\textit{iv}) we set $N = |\mathcal{V}| $ and $M = |\mathcal{E}_{el}|$. 

Next, we describe how to handle plugging -in/-out of DGUs while preserving
the stability of the mG. Whenever a DGU (say DGU $i$) wants to join
the network (e.g. DGU 6 in Figure \ref{fig:5areasplug}), it sends a
plug-in request to its future neighbors, i.e. DGUs $j\in\mathcal{N}^{el}_i$
(e.g. DGUs 1 and 5 in Figure \ref{fig:5areasplug}). Then, DGU $i$ \cite{tucci2016improved} or each DGU in
the set $\{i\}\cup \mathcal{N}^{el}_i$ \cite{tucci2015decentralized} solves an LMI problem ((24) in
\cite{tucci2016improved} and (25) in \cite{tucci2015decentralized}) that, if feasible, gives a vector gain
$K_i$ guaranteeing voltage stability in the whole mG after the
addition of DGU $i$. Otherwise, if one of the LMIs is infeasible, the
plug-in of DGU $i$ is denied and no update of matrices $K_j$,
$j\in\mathcal{N}^{el}_i$ is performed. 

The unplugging of a DGU (say DGU $m$) follows a similar procedure. It is always allowed without redesigning any local controller \cite{tucci2016improved} or it requires to successfully update, at most, controllers of DGUs $k$, $k\in\mathcal{N}^{el}_m$ before allowing the disconnection of the DGU \cite{tucci2015decentralized}. 

\begin{figure}[!htb]
	\centering
	\ctikzset{bipoles/length=.55cm}
	\tikzstyle{every node}=[font=\sffamily]
	\begin{circuitikz}[scale=0.7]
		\draw (1,1)  to [battery, o-o](1,4)
		to [short](1.5,4)
		to [short](1.5,4.5)
		to [short](3.5,4.5)
		to [short](3.5,0.5)
		to [short](1.5,0.5)
		to [short](1.5,4)
		to [short](1.5,1)
		to [short](1,1);
		\node at (2.5,2.5){ \textbf{Buck $i$}};
		\draw[-latex] (4,1.25) -- (4,3.75)node[midway,right]{$V_{ti}$};
		\draw (3.5,4) to [short](4,4)
		to [short](4.5,4)
		to [R=$R_{ti}$] (6,4)
		to [L=$L_{ti}$] (7.5,4)
		to [short, i=$\textcolor{blue}{I _{ti}}$, -] (8.5,4)
		to [short](9,4) 
		to [C, l=$C_{ti}$, -] (9,1)
		to [short](4,1)
		to [short](3.5,1);
		\draw (12.3,4)  to [R=$R_{ij}$] (14.5,4)
		to [L=$L_{ij}$] (16,4)
		to [short, -o] (16.5,4) node[anchor=north,above]{$V_j$};
		\draw (8.5,4) to (11,4) 
		to [ I ] (11 ,1)
		to [short] (9,1)
		to [short, -o] (16.5,1); 
		\draw (11,4) to [short](11.5,4);
		\draw (11,4) node[anchor=north, above]{$\textcolor{red}{V_i}$}  to [short, i_=$I_{Li}$](11,2.9);
		\node at (11,4.6)[anchor=north, above]{$PCC_i$} ;
		\draw (11,4) to [short,i<=$I_{\ell k}$] (13,4) -- (13,4); 
		\draw[black, dashed] (.5,.25) -- (12.45,.25) -- (12.45,5.5) -- (.5,5.5)node[sloped, midway, above]{{ \textbf{DGU $i$}}}  -- (.5,.25);
		\draw[black, dashed] (12.7,.25) -- (16,.25) -- (16,5.5) -- (12.7,5.5)node[sloped, midway, above]{{ \textbf{Power line $ij$}}}  -- (12.7,.25);
		\draw[red,o-] (10.9,4.15) -- (11.7,3) to (11.7,-1.5);
		\draw[red,latex-](8,-1.5)-- (9,-1.5) --  (10,-0.5)-- (11.7,-0.5);
		\draw[blue,o-latex] (8.5,4.15) to (8.5,1.75) -- (8.5,-1) to (8,-1);
		\draw (10,-2) node(a) [black, draw,fill=white!20] {$\normalsize{\int}$};
		\draw[-latex] (a.west) to (8,-2);
		\draw (11.7,-2) node(b)[ circle, draw=black, minimum size=12pt, fill=lightgray!20]{};
		\draw[red, -latex] (11.7,1.75)  -- (b.north) node[pos=0.9, left]{\textcolor{black}{\normalsize{-}}};
		\draw[latex-] (a.east) -- (b.west);
		\draw[-latex] (12.8,-2)  -- (b.east) node[pos=0.7,above]{{+}} node[pos=0.25,right]{$V{ref,i}$};
		\draw[fill=lightgray] (8,-2.25) -- (8,-0.75) -- (6.5,-1.5) -- (8,-2.25);
		\node at (7.5,-1.5) {$K_i$};
		\draw[-latex] (6.5,-1.5) -- (5.5,-1.5) -- (5.5,1.5) -- (5.5,2.5) -- (5,2.5);				   
	\end{circuitikz}
	\caption{Electrical scheme of DGU $i$ and local PnP voltage controller.}
	\label{fig:ctrl_part}
\end{figure}
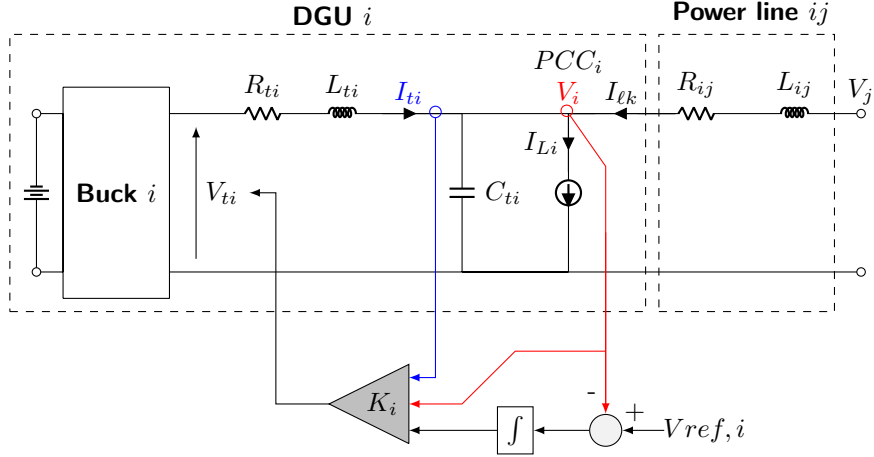

\begin{figure}
	\centering
	\includegraphics[scale=0.25]{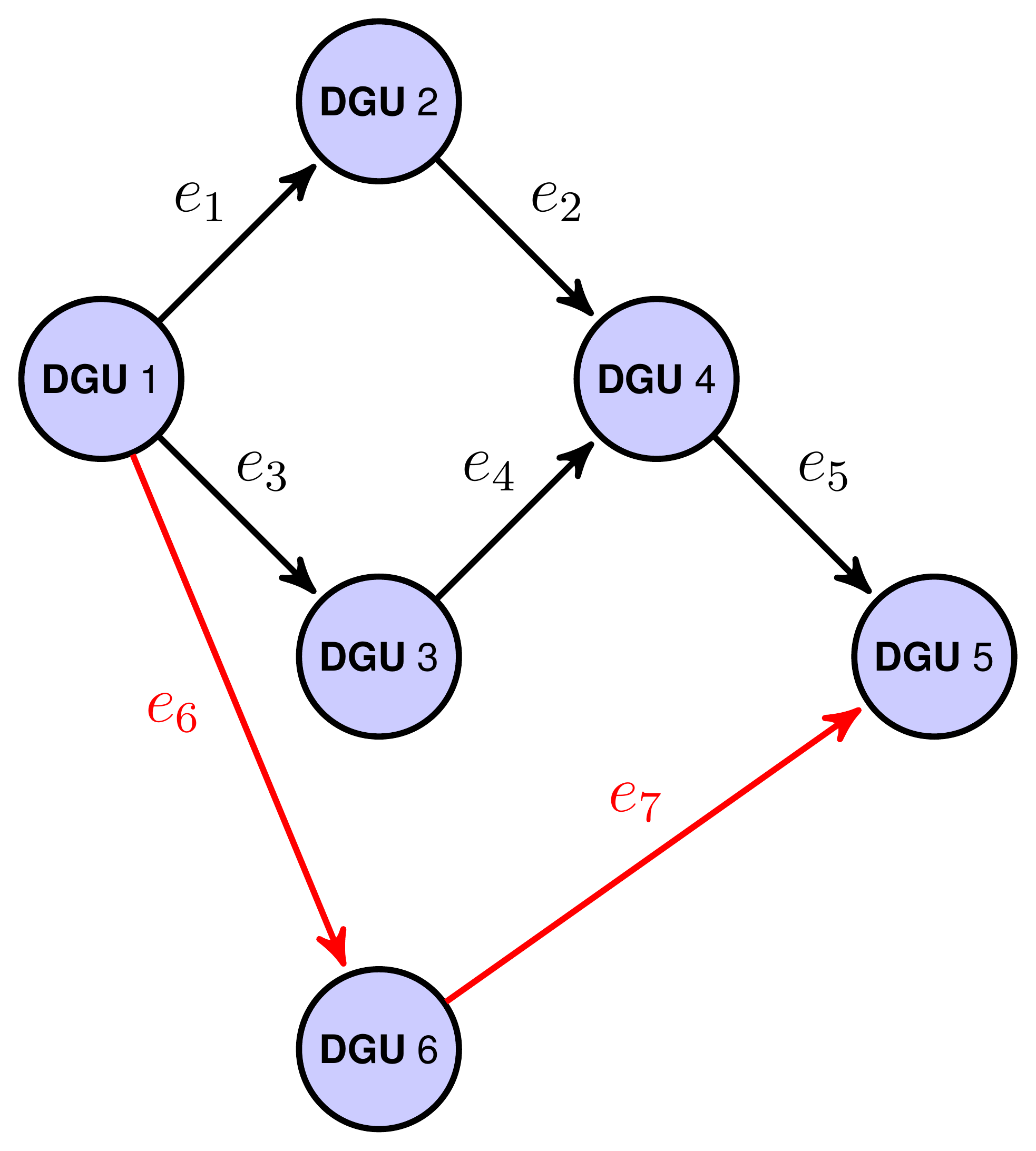}
	\caption{Graph representation of an mG composed of 5 DGUs (in black)
		and plug-in of DGU 6 (in red). Blue nodes identify DGUs whose electrical scheme is shown in Figure \ref{fig:ctrl_part}, i.e. with local loads connected to the corresponding PCCs.}
	\label{fig:5areasplug}
\end{figure}

\begin{rmk}
	\label{rmk:integral}
	Local PnP controllers can be enhanced with pre-filters so as to shape
	in a desired way the transfer function $F_{[i]}(s)$ between voltages
	$V_{ref,i}$ and $V_i$ represented in Figure
	\ref{fig:ctrl_part}. The closed-loop transfer function $F_{[i]}(s)$
	has 3 poles and in the sequel it will be approximated by a unit gain
	or a first-order system. The first approximation will be used mainly
	for tutorial reasons. The second one is very mild at low and medium
	frequencies, as can be noticed from the Bode plots of
	$F_{[i]}(s)$ in
	\cite{tucci2015decentralized}. Moreover, the presence of an
	integrator in Figure \ref{fig:ctrl_part} allows the
	voltage $V_i$ at the PCC to
	track constant references without offset when the
	disturbances (i.e. load currents $I_{Li}$) are constant.
\end{rmk}

\section{Secondary control based on consensus algorithms}
\label{sec:consensus}
\subsection{Control objectives}
Primary stabilizing  controllers, such as the PnP regulators described in Section \ref{sec:pnp_design}, have the goal of turning DGUs into controlled voltage generators, i.e. to approximate, as well as possible, the identity $V_i = V_{ref,i}$. As such, they do not ensure current sharing and voltage balancing, defined in the sequel.
\begin{definition}
	\label{defn:cs}
	For constant load currents $I_{Li},i=1,\dots,N$, \textit{current sharing} is achieved if, at steady state, the overall load current is
	proportionally shared among DGUs, i.e. if
	\begin{equation}
	\label{eq:cs_defn}
	\frac{I_{ti}}{I_{ti}^s} = \frac{I_{tj}}{I_{tj}^s}\hspace{4mm}\text{for all }i,j\in \mathcal{V},
	\end{equation}
	where $I_{ti}^s>0$ are constant scaling factors.
\end{definition}

We recall that current sharing is desirable in order to avoid
situations in which some DGUs are not able to supply local
loads, thus requiring power from other DGUs. A very common goal is to make DGUs share the total load
current proportionally to their generation capacity. This can be obtained by measuring the output currents
in per-unit (p.u.), i.e. setting each scaling factor $I_{ti}^s$ in \eqref{eq:cs_defn} equal to the corresponding DGU rated current (see Section \ref{sec:simulations} for an example).
On the other hand, if the scaling factors are all identical, the current sharing condition becomes
\begin{equation}
\label{eqn:equal_sharing}
I_{ti} = \langle\mathbf{{I}_L}\rangle\hspace{4mm}i = 1, \dots,N,
\end{equation}
where $\mathbf{I_L} = [I_{L1}, I_{L2}, \dots, I_{LN}]^T$ is the vector of the local load currents.

\begin{ass}
	\label{ass:vref}
	Voltage references are identical for all DGUs, i.e. $V_{ref,i}=V_{ref}$, $\forall i\in\mathcal{V}$.
\end{ass}
\begin{definition}
	\label{defn:vb}
	Under Assumption \ref{ass:vref}, \textit{voltage balancing} is achieved if 
	\begin{equation}
	\label{eq:vb_defn}
	\langle\mathbf{V}\rangle = V_{ref}.
	\end{equation}
	where vector $\mathbf{V} =[V_{1},V_{2}, \dots,V_{N}]^T$ collects the PCC voltages.
\end{definition}

In order to guarantee current sharing and voltage balancing, we use a consensus-based secondary
control layer, as described next. 
\subsection{Consensus dynamics}
Consensus filters are commonly employed for achieving
global information sharing or coordination through distributed
computations \cite{olfati2004consensus,lns-v.85}. In our case, as
shown in Figure \ref{fig:ctrl_complete}, we adopt the following
consensus scheme for adjusting the references of each primary voltage regulator
\begin{equation}
\label{eq:basic_consensus}
\dot {\Delta V_i}(t) =- k_{I}\sum\limits_{j=1, j\neq i} ^{N}a_{ij}\left(\frac{I_{ti}(t)}{I_{ti}^s}-\frac{I_{tj}(t)}{I_{tj}^s}\right),
\end{equation}
where $a_{ij}>0$ if DGUs $i$ and $j$ are connected by a communication
link ($a_{ij}=0$, otherwise), and the coefficient $k_I>0$ is
common to all DGUs. The use of consensus protocols has been thoroughly
studied for networks of agents with simple dynamics, e.g. simple
integrators \cite{olfati2004consensus,lns-v.85}, with the goal of
proving convergence of individual states to a common value. In our
case, however, \eqref{eq:basic_consensus} is interfaced with the mG dynamics and convergence
of currents $I_{ti}$ to the same value does not trivially follow from
standard consensus theory. This property will be rigorously analyzed in Section
\ref{sec:modeling_and_analysis}. 

In the sequel, we assume bidirectional communication, i.e. $a_{ij} =
a_{ji}$. The corresponding \textit{communication} digraph is
$\mathcal{G}_c=(\mathcal{V}, \mathcal{E}_c, W_c)$ where
$(i,j)\in\mathcal{E}_{c}\Longleftrightarrow a_{ij}> 0$ and $W_c =\text{diag}\{a_{ij}\}$. 

Considerations on the topologies of $\mathcal{G}_c$ and $\mathcal{G}_{el}$ guaranteeing stable current sharing and voltage balancing are detailed in Section \ref{sec:properties_Q}. In all cases, however, the following standing assumption must hold.
\begin{ass}
	\label{ass:conneced}
The graph $\mathcal{G}_{el}$ is weakly
		connected. The graph $\mathcal{G}_c$ is undirected and connected.
\end{ass}
From a system point of view, the collective dynamics of the group of DGUs
following \eqref{eq:basic_consensus} can be expressed as
\begin{equation}
\label{eq:ss_simpl_1}
\mathbf{\dot{\Delta V}} = -\underbrace{k_IL}_{\mathbb{L}}D\mathbf{I_{t}},
\end{equation}
where $\mathbf{\Delta V} = [\Delta V_1, \dots, \Delta V_N]^T =
\mathbf{V}-\mathbf{V_{ref}}$, $\mathbf{V_{ref}} = [V_{ref,1}, V_{ref,2}, \dots,
V_{ref,N}]^T$, $\mathbf{I_t} =[I_{t1},I_{t2}, \dots,I_{tN}]^T$, $D=\mathrm{diag}\left(\frac{1}{I_{t1}^s}, \dots, \frac{1}{I_{tN}^s}\right) = \mathrm{diag}\left(d_1, \dots, d_N\right)$ and $L  = \mathcal{L}(\mathcal{G}_c)$. Note that $\mathbb{L}$ is the Laplacian matrix of $\mathcal{G}_c$ with $W_c$ replaced by $k_IW_c$.

\section{Modeling and analysis of the complete system}
\label{sec:modeling_and_analysis}
The hierarchical control scheme of a DGU equipped with primary and
secondary regulators is shown in Figure \ref{fig:ctrl_complete}. For
studying the behavior of the closed-loop mG, we first approximate DGUs under the effect of primary controllers by unit
gains (Section \ref{sec:simpl_model}) and prove that current sharing is achieved in a stable way. We also provide conditions for voltage balancing. Results derived in this simple setting will be 
instrumental for studying the more complex scheme where primary
control loops are abstracted into first-order transfer functions
(Section \ref{sec:fo_approx}).  

\begin{figure}
	\centering
	\includegraphics[scale=0.4]{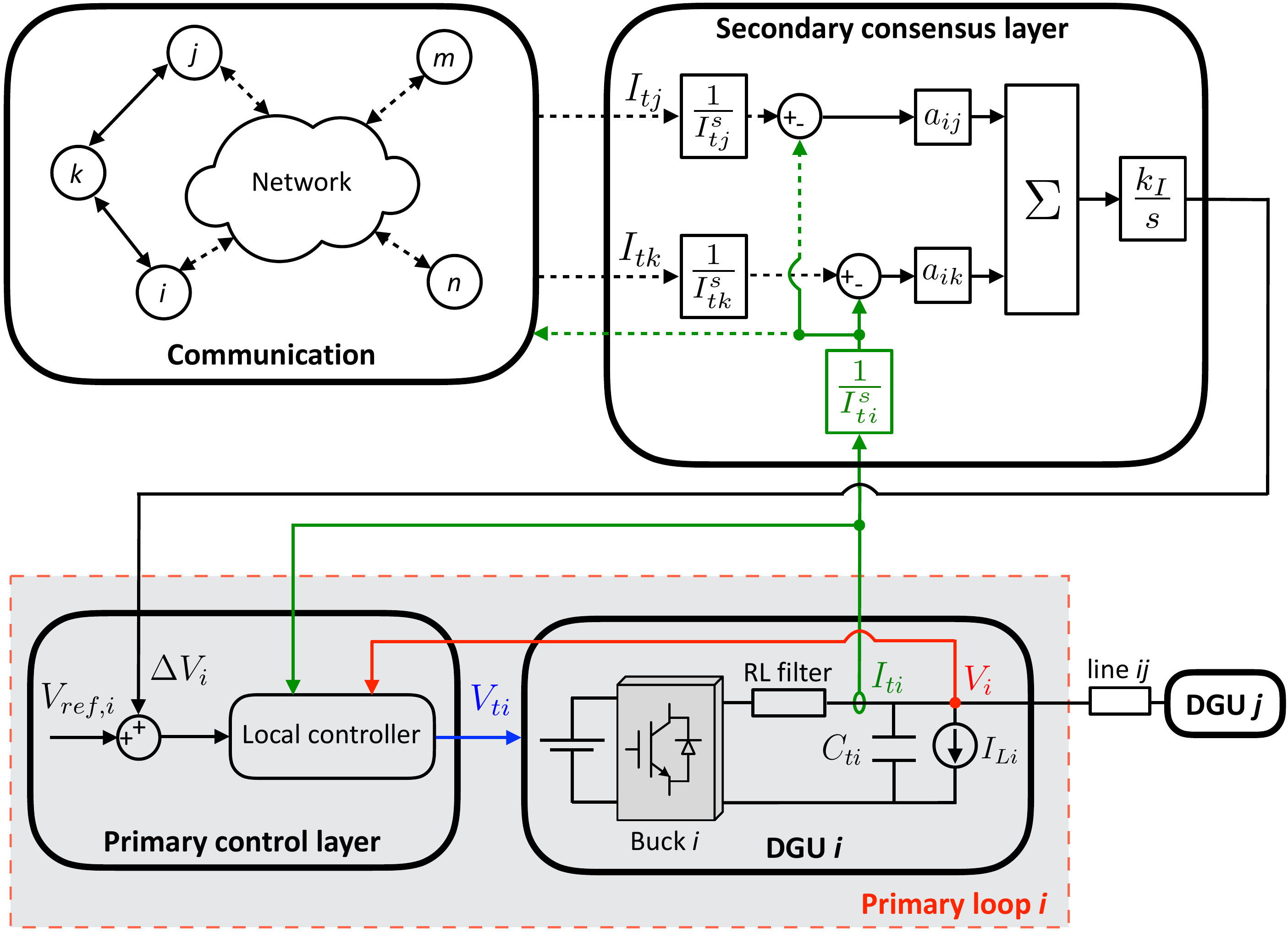}
	\caption{Complete hierarchical control scheme of DGU $i$. Subsystems $j$ and $k$ are the neighbors of DGU $i$ according to the communication graph.}
	\label{fig:ctrl_complete}
\end{figure}
\subsection{Unit-gain approximation of primary control loops}
\label{sec:simpl_model}
By approximating primary
loops with ideal unit gains, we have the relations $V_{i} =V_{ref,i}+\Delta V_{i}$,
$\forall i\in\mathcal{V}$. \\
Figures \ref{fig:hierarc}-\ref{fig:ith_ideal} show the resulting control scheme, used for deriving the dynamics of the overall mG as a function of the inputs $\mathbf{I_{L}} $ and
$\mathbf{V_{ref}}$. Starting from
the left-hand side of Figure \ref{fig:hierarc}, we have, in order, \eqref{eq:ss_simpl_1} and 
\begin{equation}
\label{eq:ss_simpl_2}
\mathbf{V} = \mathbf{\Delta V} + \mathbf{V_{ref}}.
\end{equation}
Then, from basic circuit theory, we derive the relation between the vector of voltages $\mathbf{V}$ and the vector of line currents $\mathbf{I_{\ell}} =  [I_{\ell1}, \dots,I_{\ell M}]^{T}$ as
\begin{equation}
\label{eq:ss_simpl_3}
\mathbf{I_{\ell}} = -WB_{el}^{T}\mathbf{V},
\end{equation}  
where $W$ and $B_{el}=Q(\mathcal{G}_{el})$ are the weight and the incidence matrix of $\mathcal{G}_{el}$, respectively. Next, we get
\begin{equation}
\label{eq:ss_simpl_4}
\mathbf{ I_{t}} = \mathbf{I_{L}} - B_{el}\mathbf{I_{\ell}}
\end{equation} 
and, merging equations \eqref{eq:ss_simpl_1}-\eqref{eq:ss_simpl_4}, we finally obtain
\begin{equation}
\label{eq:ss_simpl_dyn}
\begin{aligned}
\Sigma:\text{ } \mathbf{ \dot{\Delta V}} &= -\mathbb{L}D\underbrace{B_{el}W B_{el}^{T}}_{\mathbb M}\mathbf{\Delta V}-\mathbb{L}D\mathbf{I_{L}}-\mathbb{L} D\underbrace{B_{el}WB_{el}^{T}}_{\mathbb M}\mathbf{V_{ref}}\\
&= -\mathbb{Q}\mathbf{\Delta V}-\mathbb{L} D\mathbf{I_{L}}-\mathbb{Q}\mathbf{V_{ref}}
\end{aligned}
\end{equation} 
where $\mathbb M=\mathcal L(\mathcal{G}_{el}) = B_{el} W B_{el}^{T}$ is the
Laplacian matrix of the electrical network and $\mathbb{Q} =
\mathbb{L}D\mathbb{M}$. 
\subsection{Properties of the matrix $\mathbb{Q}$}
\label{sec:properties_Q}
The matrix $\mathbb{Q}$ in \eqref{eq:ss_simpl_dyn} captures the interaction of electric couplings and communication. From \eqref{eq:ss_simpl_dyn}, it governs the voltage dynamics and hence the achievement of current sharing and voltage balancing.
Notice that $\mathbb{Q}$ is obtained  pre- and post- multiplying a diagonal matrix by a Laplacian ($\mathbb{L}$ and $\mathbb{M}$, respectively). It follows that $\mathbb{Q}$ is not a Laplacian matrix itself because it might fail to be symmetric and have
positive off-diagonal entries, even if weights of $\mathcal{G}_{el}$ and $\mathcal{G}_{c}$
are positive. Nevertheless, in the sequel we provide two distinct conditions under which $\mathbb{Q}$ preserves some key features of Laplacian matrices. Before proceeding, we introduce the following preliminary result.
\begin{figure}[!htb]
	\centering
	\begin{subfigure}[!htb]{0.8\textwidth}
		\centering
		\includegraphics[scale=0.7,width = \columnwidth]{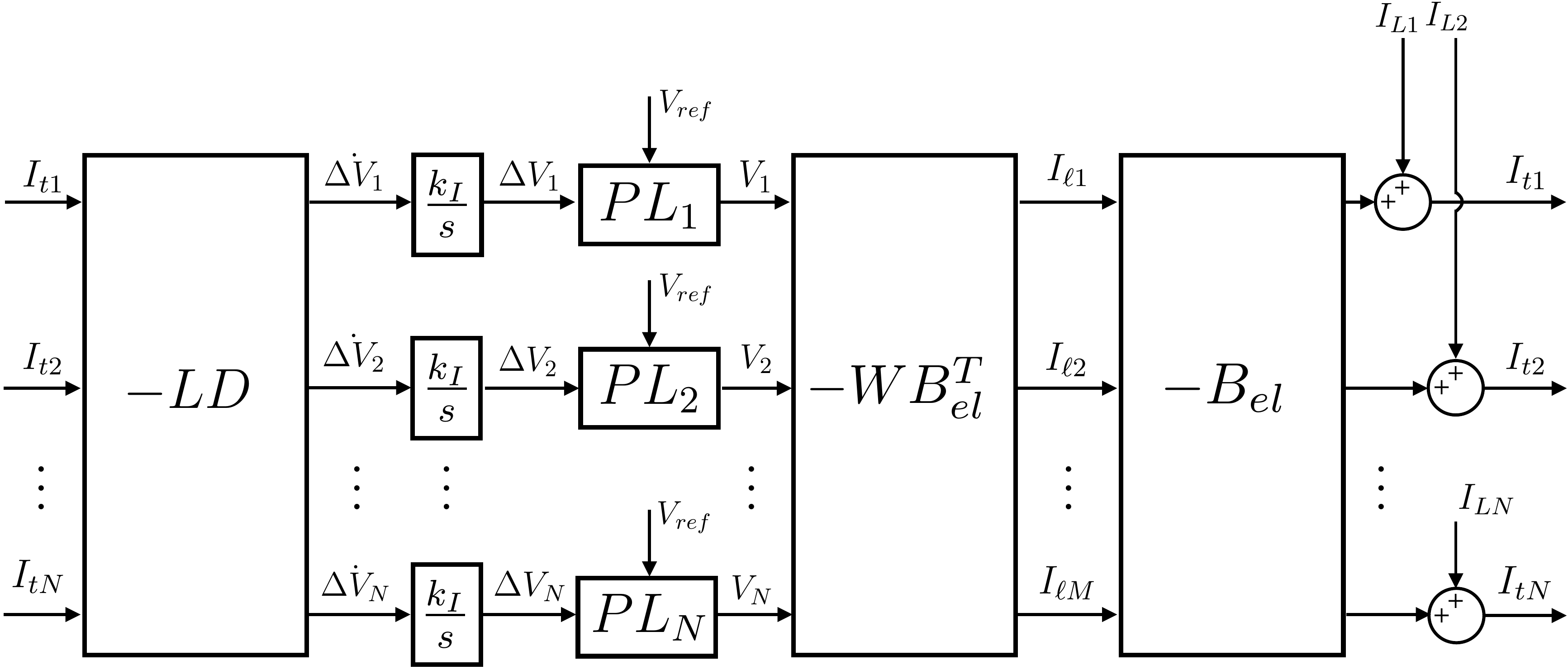}
		\caption{Hierarchical control scheme.}
		\label{fig:hierarc}
	\end{subfigure}
	\begin{subfigure}[!htb]{0.35\columnwidth}
		\centering
		\includegraphics[scale=0.35, width = \columnwidth]{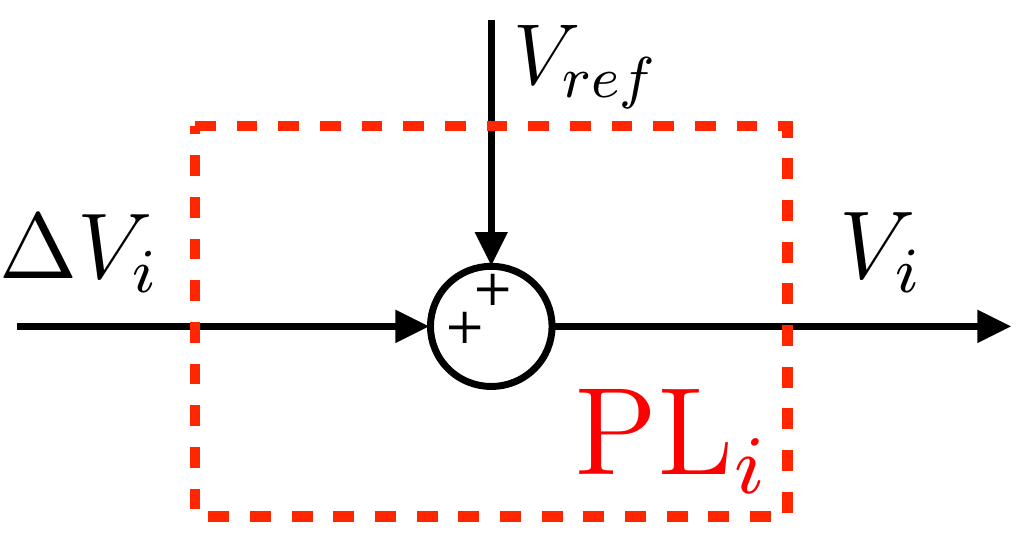}
		\caption{$i$-th ideal primary loop.}
		\label{fig:ith_ideal}
	\end{subfigure}
	\begin{subfigure}[!htb]{0.45\columnwidth}
		\centering
		\includegraphics[scale=0.4, width = \columnwidth]{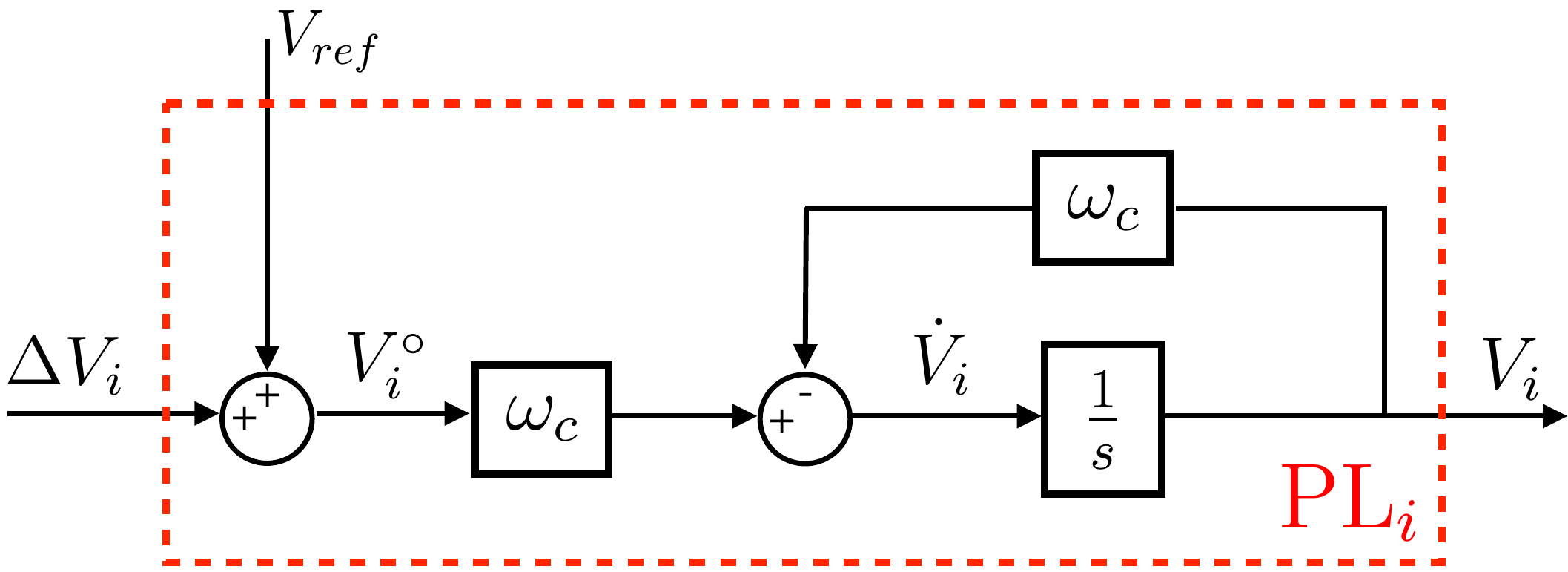}
		\caption{First-order approximation of the $i$-th primary loop.}
		\label{ith_first_order}
	\end{subfigure}
	\caption{Hierarchical control scheme and
		Primary Loop (PL) approximations.}
	\label{fig:hierarc_scheme}
\end{figure}
\begin{prp}
	\label{prp:projection}
	It holds $P_{H^1}(D\mathbb{M} H^1) = H^1$.
\end{prp}
\begin{proof}
	From Proposition \ref{pr:laplacian_prop}-(\ref{lapl_4}), we know that $\mathbb{M}(H^1|H^1)$ is invertible, hence surjective. Then,
	\begin{equation*}
	D\mathbb{M} H^1 = D H^1.
	\end{equation*}
	We now study the projection map ${P}_{H^1}:DH^1\rightarrow H^1$. Since $D$ is invertible, we have that $\mathrm{dim}(D H^1) = \mathrm{dim}(H^1) = N-1$. Therefore, one has
	\begin{equation}
	\label{eq:rank_nullity_projection}
	\mathrm{dim}(\mathrm{Range}({{P}_{H^1}})) + 	\mathrm{dim}(\mathrm{Ker}({{P}_{H^1}})) = N-1.
	\end{equation}
	The next step is to show that 
	\begin{equation}
	\label{eq:ker_projection}
	\mathrm{Ker}({P}_{H^1}) = \{0\}, 
	\end{equation}
	so that, from \eqref{eq:rank_nullity_projection}, ${P}_{H^1}$ is surjective. Let $u\in DH^1$ and set $u = \bar{u} + \hat{u}$, where $\bar{u}\in H^1_{\perp}$ and $\hat{u} \in H^1$. Hence, ${P}_{H^1}u = {P}_{H^1}\hat{u} = \hat{u}$ and, if $u\in\mathrm{Ker}({P}_{H^1})$, then $\hat{u}=0$. In other words, $u\in\mathrm{Ker}({P}_{H^1})$ verifies
	\begin{equation*}
	\label{eq:ler_projection}
	\quad\left\lbrace
	\begin{aligned}
	u &= \alpha_u\mathbf{1}\hspace{4mm}\text{for some }\alpha_u\in\mathbb{R}\\
	u &= Dv\hspace{5.5mm}\text{for some }v\in H^1\\
	\end{aligned}
	\right.
	\end{equation*}
	One has
	\begin{equation*}
	||u||^{2}_{D^{-1}} = u^T D^{-1}u = \alpha_{u}\mathbf{1}^T D^{-1}D v = \alpha_{u}\mathbf{1}^T v = 0
	\end{equation*}
	where the last identity follows from the fact that $v\in H^1$ has zero average. But, since $||\cdot||_{D^{-1}}$ is a norm, $||u||^{2}_{D^{-1}}$ implies $u=0$. This shows \eqref{eq:ker_projection}.
\end{proof}

	At this point, we can introduce two assumptions which allow us to characterize the eigenstructure of matrix $\mathbb{Q}$; then, we discuss their impact on the choice of the communication graph topology as well as on the value of coefficients $a_{ij}$ in \eqref{eq:basic_consensus}.
	\begin{ass}
		\label{ass:D_identity}
		The diagonal matrix containing the scaling factors for the output currents coincides with the identity, i.e. $D=I$.  
	\end{ass}
	\begin{ass}
		\label{ass:same_topology}
		It holds $D\neq I$, $D > 0$ and the product $\mathbb{L}D\mathbb{M}$ commutes (i.e. $\mathbb{L}D\mathbb{M} =  \mathbb{M}D\mathbb{L}$).
	\end{ass}
	\begin{rmk}
		Under Assumption \ref{ass:D_identity}, the desired goal is that all converters in the mG produce the same current (measured in Ampere and not in p.u.), i.e. that \eqref{eqn:equal_sharing} holds. This can be desirable, for instance, when all the converters in the network have the same generation capacity. We notice that Assumption \ref{ass:D_identity} does not enforce constraints on the graphs $\mathcal{G}_{el}$ and $\mathcal{G}_{c}$, which must fulfill Assumption \ref{ass:conneced} only. In particular, the topology of $\mathcal{G}_{c}$ can be chosen based on considerations on the network technology and for optimizing the performance of the secondary control layer, independently of the electrical interactions.
	\end{rmk}
	\begin{rmk}
		\label{rmk:same_topology}
		Assumption \ref{ass:same_topology} is suited to the case of converters with different ratings. However, it requires matrix $\mathbb{Q}$ to be symmetric. An interesting case where this condition is verified is given by 
		\begin{equation}
		\label{eq:prop_laplacians}
		\mathbb{L} =  \mu\mathbb{M},
		\end{equation}
		where $\mu>0$ is a global parameter, common to all the DGUs\footnote{From \eqref{eq:prop_laplacians}, $\mathbb{Q} = \mathbb{L}D\mathbb{M} = \mu\mathbb{M}D\mathbb{M}$, which commutes since $\mathbb{M}$ and $D$ are symmetric and $\mu$ is a scalar.}. Relation \eqref{eq:prop_laplacians} holds if the following conditions are simultaneously guaranteed:
		\begin{enumerate}[(a)]
			\item $\mathcal{G}_{el}$ and $\mathcal{G}_{c}$ have the same topology;
			\item \label{rmk:cond_b}coefficients in \eqref{eq:basic_consensus} are chosen as $a_{ij} = \mu\frac{1}{R_{ij}}$ if DGUs $i$ and $j$ are connected by a communication link.
		\end{enumerate} 
		Condition \eqref{rmk:cond_b} has an impact on the design of local Laplacian control laws \eqref{eq:basic_consensus}. Notably, each agent (say DGU $i$) needs to know the global parameter $\mu$ and the value of the conductances $\frac{1}{R_{ij}}$ connecting it to its electrical neighbors $j$ belonging to set $\mathcal{N}^{el}_i$ (which, in this particular case, coincides with $\mathcal{N}^{c}_i$).
	\end{rmk}
	\begin{rmk}
		Assumptions \ref{ass:D_identity} or \ref{ass:same_topology} are instrumental in guaranteeing asymptotic stability of the hierarchical control architecture in Figure \ref{fig:ctrl_complete}. Indeed, as will be shown in the following Proposition, the above conditions imply that the eigenvalues of $\mathbb{Q}$ are nonnegative.
		Note that, if neither Assumption \ref{ass:D_identity} nor \ref{ass:same_topology} are verified, $\mathbb{Q}$ can have negative eigenvalues (see the example in Appendix \ref{app:eig_Q}).
	\end{rmk}

\begin{prp}
	\label{prop_2}
	The matrix $\mathbb{Q} = \mathbb{L}D\mathbb{M}$ has the following properties:
	\begin{enumerate}[(i)]
		\item \label{prop:zero_row_sum} $\mathrm{Ker}(\mathbb{Q})=H^1_{\perp}$;
		\item \label{prop:range} $\mathrm{Range}(\mathbb{Q}) = H^1$;
		\item \label{prop:Q_invertible} the linear transformation $\mathbb{Q}(H^1|H^1)$ is invertible;
		\item \label{prop:diag_and_psd} under Assumption \ref{ass:D_identity} or \ref{ass:same_topology}, $\mathbb{Q}$ is diagonalizable and has real nonnegative eigenvalues.
		\item \label{prop:zero_eig} under Assumption \ref{ass:D_identity} or \ref{ass:same_topology}, the zero eigenvalue of $\mathbb{Q}$ has algebraic multiplicity equal to one.
	\end{enumerate}
\end{prp}
\begin{proof}
	We start by proving point (\ref{prop:range}). Since, from Proposition \ref{pr:laplacian_prop}-(\ref{lapl_3}), $\mathrm{Ker}(\mathbb{M}) = H_{\perp}^1$, one has $\mathbb{M}\mathbb{R}^{N}=\mathbb{M}(H_{\perp}^1\oplus H^1) = \mathbb{M} H^1$. Furthermore, from Proposition \ref{prp:projection}, ${P}_{H^1}(D\mathbb{M} H^1) = H^1$. Proposition \ref{pr:laplacian_prop}-(\ref{lapl_4}) applied to the Laplacian $\mathbb{L}$ shows that $\mathbb{L}(D\mathbb{M}\mathbb{R}^N) = H^1$, which is (\ref{prop:range}).
	
	For proving point (\ref{prop:zero_row_sum}), we recall that $\mathrm{Ker}(\mathbb{M}) = H_{\perp}^1$ and then $\mathrm{Ker}(\mathbb{L} D \mathbb{M})\supseteq H_{\perp}^1$. From (\ref{prop:range}) we have that $\mathrm{dim}(\mathrm{Range}(\mathbb{Q})) = N-1$. The equation $\mathrm{dim}(\mathrm{Range}(\mathbb{Q}))+\mathrm{dim}(\mathrm{Ker}(\mathbb{Q})) = N$ implies that $\mathrm{dim}(\mathrm{Ker}(\mathbb{Q})) = 1$. Since $\mathrm{dim}(H_{\perp}^1)=1$, we have $\mathrm{Ker}(\mathbb{Q}) = H_{\perp}^1$.
	
	In order to prove point (\ref{prop:Q_invertible}), we  show that $\mathbb{Q}(H^1|H^1)$ is both surjective and injective \cite[p.~50]{lang1987linear}. The surjectivity of $\mathbb{Q}$ on $H^1$ has been shown above when proving point (\ref{prop:range}). For proving the injectivity, we need to check if it holds
	\begin{equation*}
	\label{eq:injectivity_test}
	\forall b\in H^1\hspace{3mm}\forall x,y \in H^1\hspace{3mm}(\mathbb{Q}x=b \mbox{ and } \mathbb{Q}y=b)\Rightarrow x=y.
	\end{equation*}
	Now, $\mathbb{Q} x =\mathbb{Q} y = b$ implies that $\mathbb{Q}(x-y) = 0$. It means that $ x-y\in \mathrm{Ker}(\mathbb{Q})= H^1_{\perp}$, therefore $\exists\alpha\in{\mathbb{R}}$ such that $x-y = \alpha \mathbf{1}_n$. However, since $x-y\in H^1$, $x-y = \alpha \mathbf{1}_n$ is verified only for $\alpha = 0$; this leads to $x = y$. 
	
	As regards statement (\ref{prop:diag_and_psd}), we first consider the case in which Assumption \ref{ass:D_identity} holds. Since $D = I$, we have that $\mathbb{Q} = \mathbb{L}\mathbb{M}$. Hence, $\mathbb{Q}$ is the product of two matrices, both positive semidefinite in the real sense. Moreover, since $\mathbb{L}$ and $\mathbb{M}$ are symmetric, they are positive semidefinite also in the complex sense
		\cite{pease1965methods}. The proof concludes by applying Corollary 2.3 in \cite{hong1991jordan}, which states that the product of two complex positive
		semidefinite matrices is diagonalizable and has nonnegative real
		eigenvalues. 
	
	We now prove statement \eqref{prop:diag_and_psd} when Assumption \ref{ass:same_topology} holds. Since $D$ is diagonal with positive elements, the matrix $D^{\frac{1}{2}}$ verifying $D = D^{\frac{1}{2}}D^{\frac{1}{2}}$ exists and is invertible. Then, $\mathbb{Q}$ can be written as follows
		\begin{equation}
		\label{eq:rewrite_Q}
		\mathbb{Q} = D^{-\frac{1}{2}}\underbrace{D^{\frac{1}{2}}\mathbb{L} D^{\frac{1}{2}}}_{\mathcal{L}}\underbrace{D^{\frac{1}{2}}\mathbb{M} D^{\frac{1}{2}}}_{\mathcal{M}}D^{-\frac{1}{2}}.
		\end{equation}
		Matrices $\mathcal{L}$ and $\mathcal{M}$ in \eqref{eq:rewrite_Q} are positive semidefinite in the real sense and symmetric; hence, they are positive semidefinite also in the complex sense. Therefore, also in this case, we can use Corollary 2.3 in \cite{hong1991jordan} to state that  $\mathcal{L}\mathcal{M}$ is has nonnegative real eigenvalues. Now, since $D^{-\frac{1}{2}}$ in \eqref{eq:rewrite_Q} is symmetric, matrix $\mathcal{L}\mathcal{M}$ is congruent to $\mathbb{Q}$. Thus, since under Assumption \ref{ass:same_topology} $\mathcal{L}\mathcal{M}$ is symmetric\footnote{Indeed, $(\mathcal{L}\mathcal{M})^T = \mathcal{M}\mathcal{L} = D^{\frac{1}{2}}\underbrace{\mathbb{M} D^{\frac{1}{2}}D^{\frac{1}{2}}\mathbb{L}}_{\mathbb{M}D\mathbb{L} = \mathbb{L}D\mathbb{M}} D^{\frac{1}{2}} = \mathcal{L}\mathcal{M}$.}, by Sylvester's law of inertia \cite{horn2012matrix}, the inertia of $\mathbb{Q}$ and $\mathcal{L}\mathcal{M}$ coincide, i.e.
		\begin{equation*}
		\label{eq:same_inertia}
		(i_{+}(\mathcal{L}\mathcal{M}), 0, i_{0}(\mathcal{L}\mathcal{M})) = (i_{+}(\mathbb{Q}), 0, i_{0}(\mathbb{Q})).
		\end{equation*} 
		This concludes the proof of statement \eqref{prop:diag_and_psd} under Assumption \ref{ass:same_topology}.
		
		Finally, point \eqref{prop:zero_eig} follows from points \eqref{prop:zero_row_sum} and \eqref{prop:diag_and_psd}. Indeed, since $\Qset$ is diagonalizable, the algebraic and geometric multiplicity of null eigenvalues coincide. From point (\ref{prop:zero_row_sum}), since $\mathrm{dim}(H_{\perp}^1)=1$, we conclude that the zero eigenvalue of $\Qset$ is unique.
\end{proof}

\subsubsection{Analysis of equilibria}
\label{sec:eq_an_1}
In order to evaluate the steady-state behavior of the electrical
signals appearing in Figure \ref{fig:hierarc}-\ref{fig:ith_ideal},
we study the equilibria of system \eqref{eq:ss_simpl_dyn}. Hence, for given constant inputs $(\mathbf{I_L^*},\mathbf{{V}_{ref}^*})$, we characterize the solutions $\mathbf{{\Delta V}^{*}}$ of equation
\begin{equation}
\label{eq:eq_simpl}
\mathbb{Q}\mathbf{{\Delta V}^{*}}= -\mathbb{L}D\mathbf{{I_{L}^*}}-\mathbb{Q}\mathbf{{V}_{ref}^*}
\end{equation}
through the following Proposition.
\begin{prp}
	\label{prop:equilibrium}
	For equation \eqref{eq:eq_simpl},
	\begin{enumerate}[(i)]
		\item \label{st_2} there is only one solution $\mathbf{\widetilde{ \Delta V}^{*}}\in H^1$;
		\item \label{st_3} all solutions $\mathbf{{\Delta V}^{*}}\in\mathbb{R}^{N}$ can be written as
		\begin{equation}
		\label{eq:infinite_equilibria}
		\mathbf{{\Delta V}^{*}} = \mathbf{\widetilde{\Delta V}^{*}} + \alpha\mathbf{1}_N\hspace{5mm}\alpha\in\mathbb{R}.
		\end{equation}
	\end{enumerate}
\end{prp}
\begin{proof}
	Proposition \ref{prop_2}-(\ref{prop:range}) shows that
	\eqref{eq:eq_simpl} has solutions only if $-\mathbb{L}D
	\mathbf{{I_{L}^*}}-\mathbb{Q}\mathbf{{V}_{ref}^*}\in H^1$. From
	Propositions \ref{pr:laplacian_prop}-(\ref{lapl_3}) and
	\ref{prop_2}-(\ref{prop:range}) this is always true. Statement (\ref{st_2}) directly follows from Proposition \ref{prop_2}-(\ref{prop:Q_invertible}).\\
	For the proof of statement (\ref{st_3}), we split $\mathbf{\Delta
		V}^{*}\in\mathbb{R}^N$ as in \eqref{eq:decomposition}, i.e. $\mathbf{\Delta
		V}^{*} = \mathbf{\widehat{\Delta V}^{*}} + \mathbf{\overline{\Delta
			V}^{*}}$. From \eqref{eq:eq_simpl} and Proposition
	\ref{prop_2}-(\ref{prop:zero_row_sum}), one has that, irrespectively
	of $\mathbf{\overline{\Delta V}^{*}}=\alpha \mathbf{1}_N\in H_{\perp}^1$, $\mathbb{Q} \mathbf{\widehat{\Delta V}^{*}}  = -\mathbb{L}D
	\mathbf{{I_{L}^*}}-\mathbb{Q}\mathbf{{V}_{ref}^*}$. Moreover, from the first part of the proof, it holds $\mathbf{\widehat{\Delta V}^{*}}=\mathbf{\widetilde{\Delta V}^{*}}$.
\end{proof}
Next, we relate properties of
the equilibria of \eqref{eq:ss_simpl_dyn} to current sharing and voltage balancing.
\begin{prp}
	\label{prop:volt_bal}
	Consider system \eqref{eq:ss_simpl_dyn} with constant inputs
	$(\mathbf{I_{L}^*},\mathbf{ V_{ref}^*})$. Then, current sharing is achieved at steady state. Moreover, if $\mathbf{V_{ref}^*} =
	V_{ref}\mathbf{1}_N$ (i.e. if Assumption \ref{ass:vref} holds) and $\alpha$
	in \eqref{eq:infinite_equilibria} is equal to zero, then the
	equilibrium $\mathbf{V^{*}}$ verifies the voltage balancing condition \eqref{eq:vb_defn}.
\end{prp}
\begin{proof}
	At the equilibrium, since $\mathrm{Ker}(\mathbb{L}) = H_{\perp}^1$ (see Proposition \ref{pr:laplacian_prop}-\eqref{lapl_3}), from \eqref{eq:ss_simpl_1} we have that 
	\begin{equation}
	\label{eq:I_curr_sharing}
	\begin{aligned}
	-\mathbb{L} D \mathbf{I_t^{*}}=\mathbf{0}_N &\Leftrightarrow D\mathbf{{I}_t^*} = \bar{I_{t}}\mathbf{1}_N\\
	&\Leftrightarrow  \left[
	\frac{I_{t1}^*}{I_{t1}^s},\dots,
	\frac{I_{tN}^*}{I_{tN}^s}\right]^T =\bar{I_{t}}\mathbf{1}_N,
	\end{aligned}
	\end{equation}
	which is \eqref{eq:cs_defn}. 
	Let now $\mathbf{{\Delta V}^{*}}$ be an equilibrium for system \eqref{eq:ss_simpl_dyn}. Replacing \eqref{eq:infinite_equilibria} in \eqref{eq:ss_simpl_2} and averaging the resulting vector, we get
	$ \langle \mathbf{V^*}\rangle = \underbrace{\langle \mathbf{\widetilde{\Delta V}^*}\rangle}_{=0} +\underbrace{N\alpha}_{=0}+V_{ref},$ which is \eqref{eq:vb_defn}.
\end{proof}
\subsubsection{Stability analysis}
\label{sec:st_an_1}
The similarities established in Proposition \ref{prop_2} between the spectral properties of graph Laplacians and the matrix $\mathbb{Q}$ will allow us to study the stability properties of \eqref{eq:ss_simpl_dyn} using methods similar to the ones adopted for analysis of classical consensus dynamics. Results in this section follow the approach in \cite{1643380}, where stability of consensus is analyzed through the use of invariant subspaces. An advantage of this rationale is that it carries over almost invariably to the case of more complex models of primary loops (Section \ref{sec:fo_approx}).

In the sequel, we prove exponentially stable convergence of
$\mathbf{\Delta V}$ in \eqref{eq:ss_simpl_dyn} to an equilibrium
ensuring both current sharing and voltage balancing for constant
$\mathbf{I_L}^*$ and $\mathbf{V^*_{ref}} = V_{ref}\mathbf{1}_N$. We first show that projections $P_{H^1_{\perp}}(\mathbf{\Delta V}) = \mathbf{\overline{\Delta V}}$ and $P_{H^1}(\mathbf{\Delta V}) = \mathbf{\widehat{\mathbf{\Delta V}}}$ have non-interacting dynamics (or, equivalently, that subspaces $H^1$ and $H^1_{\perp}$ are invariant for \eqref{eq:ss_simpl_dyn}).
\begin{prp}
	\label{pr:invariant}
	If $\mathbf{\Delta V}$ is given by system $\Sigma$ in
	\eqref{eq:ss_simpl_dyn} for $\mathbf{\Delta V}(0) = \mathbf{ \Delta
		V}_0$, then $\mathbf{\Delta V} = \mathbf{\overline{\Delta V}} + \mathbf{\widehat{\Delta V}}$, where $\mathbf{\overline{\Delta V}}\in H_{\perp}^1$ and $\mathbf{\widehat{\Delta V}}\in H^1$ fulfill
	\begin{equation}
	\label{eq:subsysDGUi}
	\text{ $\overline\Sigma:$}\left\lbrace
	\begin{aligned}
	\mathbf{ {\dot{\overline{\Delta V}}}} &= \mathbf{0}_N\\
	\mathbf{\overline{ \Delta V}}(0) &=\mathbf{\langle{\Delta V_0}\rangle}\mathbf{1}_N
	\end{aligned}
	\right.
	\end{equation} and
	\begin{equation}
	\label{eq:simpl_initial_state_H1}            
	\text{ $\widehat\Sigma:$}\left\lbrace
	\begin{aligned}
	\mathbf{\dot{\widehat{ \Delta V}}} &= -{ \mathbb{Q}}\mathbf{\widehat{\Delta V}}-{\mathbb{L}}D\mathbf{{I_{L}}}-{\mathbb{Q}}\mathbf{{V_{ref}}}\\
	\mathbf{\widehat{ \Delta V}}(0) &=\mathbf{{ \Delta V}}_0 - \mathbf{\overline{ \Delta V}}_0.
	\end{aligned}
	\right.
	\end{equation}
\end{prp}
\begin{proof}
	We write vectors $\mathbf{ \Delta V}(0)$, $\mathbf{I_{L}}$ and
	$\mathbf{V_{ref}}$ according to the decomposition
	\eqref{eq:decomposition}, i.e. using
	``\hspace{1mm}$\widehat{\textcolor{white}{v}}$\hspace{1mm}" and
	``\hspace{1mm}$\overline{\textcolor{white}{v}}$\hspace{1mm}" for
	denoting their $H^1$ and $H^1_{\perp}$ components, respectively. As described in \cite{1643380}, we analyze the dynamics of
	$\mathbf{\overline{\Delta V}}$ by averaging both sides of
	\eqref{eq:ss_simpl_dyn} and $\mathbf{\Delta V}(0) = \mathbf{ \Delta
		V}_0$, respectively. From points (\ref{prop:zero_row_sum})-(\ref{prop:range}) of Proposition
	\ref{prop_2}, we have
	$\langle-\mathbb{Q}\mathbf{\Delta V}\rangle = 0$ and $\langle-\mathbb{Q}\mathbf{
		V_{ref}}\rangle=0$. Since $\mathrm{Range}(\mathbb{L}) = H^1$ (see Proposition \ref{pr:laplacian_prop}-(\ref{lapl_3})), we also have $\langle-\mathbb{L} D\mathbf{I_L}\rangle=0$, hence obtaining $\frac{d}{dt}\langle\mathbf{\Delta V}\rangle = 0$. Recalling that $\mathbf{\overline{\Delta V}_0} =\mathbf{\langle{\Delta V_0}\rangle}\mathbf{1}_N$, we obtain \eqref{eq:subsysDGUi}.
	
	Next, we analyze $\mathbf{\widehat{\Delta V}}=\mathbf{\Delta V}-\mathbf{\overline{\Delta V}}$. We have
	\begin{equation*}
	\label{eq:splitting_H1}
	\mathbf{\dot{\widehat{\Delta V}}} = \mathbf{\dot{\Delta V}}-\underbrace{\mathbf{\dot{\overline{{\Delta V}}}}}_{=\mathbf{0}_N} = -{ \mathbb{Q}}\mathbf{{\Delta V}}-{\mathbb{L}}D\mathbf{{I_{L}}}-{\mathbb{Q}}\mathbf{{V_{ref}}}
	\end{equation*}
	and $\mathbf{\widehat{ \Delta V}}(0) =\mathbf{{ \Delta V}}_0 - \mathbf{\overline{ \Delta V}}_0$. From Proposition \ref{prop_2}-(\ref{prop:zero_row_sum}), $\mathbb{Q}\mathbf{{\Delta V}}=\mathbb{Q}\mathbf{\widehat{\Delta V}}$ and then we have \eqref{eq:simpl_initial_state_H1}.  
\end{proof}
\begin{rmk}
	\label{rmk:splitting}
	The splitting of $\Sigma$ into systems $\overline\Sigma$ and
	$\widehat\Sigma$ implies that, if $\mathbf{\Delta V_0}$ has zero
	average, then $\mathbf{\Delta V}(t)$ has the same property,
	$\forall t\geq 0$ and irrespectively of inputs
	$(\mathbf{I_L},\mathbf{V_{ref}})$. This behavior can be realized by
	suitable initialization of the integrators appearing in Figure \ref{fig:hierarc}.
\end{rmk}
According to system $\overline\Sigma$, the value of $P_{H^1_{\perp}}(\mathbf{\Delta V})=\mathbf{\overline{\Delta V}}$ remains constant over time and equal to $\mathbf{\overline{\Delta V_0}}$. Hence, in order to characterize the stability of equilibria \eqref{eq:infinite_equilibria}, it is sufficient to study the dynamics \eqref{eq:simpl_initial_state_H1}. In an equivalent way, one can consider system \eqref{eq:ss_simpl_dyn} and the following definition of stability on a subspace.
\begin{definition}
	Let $\mathcal{V}$ be a subspace of $\mathbb{R}^n$. The origin of $\dot x = \mathcal{A} x$, $x(t)\in\mathbb{R}^n$ is Globally Exponentially Stable (GES) on $\mathcal{V}$ if $\exists\kappa, \eta>0$ : $\lVert P_{\mathcal{V}}x(t) \rVert\leq\kappa e^{-\eta t}\lVert P_{\mathcal{V}}x(0) \rVert$. The parameter $\eta$ is termed \textit{rate of convergence}.
\end{definition}

Note that $\Sigma$ is a linear system and, for stability analysis, we can neglect inputs, hence obtaining 
\begin{equation}
\label{eq:no_inputs}            
\left\lbrace
\begin{aligned}
\mathbf{\dot{{ \Delta V}}} &= -{ \mathbb{Q}}\mathbf{{\Delta V}}\\
\mathbf{{ \Delta V}}(0) &=\mathbf{{ \Delta V}}_0.
\end{aligned}
\right.
\end{equation}
\begin{theorem}
	\label{thm:GES_H1}
	Under Assumption \ref{ass:D_identity} or \ref{ass:same_topology}, the origin of \eqref{eq:no_inputs} is GES on $H^1$. Moreover, the rate of convergence is the smallest strictly positive eigenvalue of $\mathbb{Q}$.
\end{theorem}
\begin{proof}
	The proof of Theorem \ref{thm:GES_H1} is presented in Appendix \ref{appendix_1}.
\end{proof}
The above results reveal that, given an initial
condition $\mathbf{{ \Delta V}}(0)=\mathbf{\Delta V_0}$ for system \eqref{eq:ss_simpl_dyn}
and constant inputs $\mathbf{I^*_L}$ and
$\mathbf{V_{ref}^*}=V_{ref}\mathbf{1}_N$, the state $\mathbf{\Delta V}$
converges to the equilibrium \eqref{eq:infinite_equilibria} with $\alpha = \mathbf{\langle{\Delta V_0}\rangle}$.\\
Summarizing the main results of this Section, we have that the consensus scheme described by \eqref{eq:basic_consensus}, Assumption \ref{ass:vref} and 
\begin{equation}
\label{eq:cs_vb}
\mathbf{\langle{\Delta V_0}\rangle}=0
\end{equation}
guarantee the asymptotic achievement of current sharing and voltage balancing in a GES fashion. 

\subsection{First-order approximation of primary control loops}
\label{sec:fo_approx}
Figure \ref{fig:hierarc} and Figure \ref{ith_first_order} show the overall
closed-loop scheme of an mG equipped with (\textit{i}) consensus current loops
and (\textit{ii}) primary control loops modeled as first-order transfer
functions. Different from the case analyzed in Section
\ref{sec:simpl_model}, each local dynamics is now described by means
of two states which are the state of the consensus current loop
($\Delta V_i$) and the state of the controlled DGU ($V_{i}$ in Figure
\ref{ith_first_order}). We highlight that relations
\eqref{eq:ss_simpl_1} and \eqref{eq:ss_simpl_2} still hold, while the additional state equation is
\begin{equation}
\label{eq:state2_1st_order}
\mathbf{\dot V} = \Omega\mathbf{V^{\circ}}-\Omega\mathbf{V},
\end{equation}
where vectors
$\mathbf{V^{\circ}}=[V^{\circ}_1,V^{\circ}_2,\dots,V^{\circ}_N]^T$ and $\mathbf{V}$ belong to $\mathbb{R}^N$, and the diagonal matrix $\Omega = \omega_c I\in\mathbb{R} ^{N\times N}$, $\omega_c>0$, collects on its diagonal the approximate bandwidth of each controlled DGU.
In view of Remark \ref{rmk:integral}, assuming equal approximate bandwidths for all the controlled DGUs is a mild constraint.\\
As in Section \ref{sec:simpl_model}, in order to find the dynamics of the closed-loop scheme, we write relations among mG variables. From Figure \ref{fig:hierarc}, we notice that \eqref{eq:ss_simpl_1} holds, and
\begin{equation}
\label{eq:v_circ}
\mathbf{V^{\circ}}= \mathbf{\Delta V} +\mathbf{V_{ref}}.
\end{equation}
Always from Figure \ref{fig:hierarc}, we have that, for line and output currents, equations \eqref{eq:ss_simpl_3} and \eqref{eq:ss_simpl_4} are still valid. By merging relations \eqref{eq:ss_simpl_1}, \eqref{eq:v_circ}, \eqref{eq:state2_1st_order}, \eqref{eq:ss_simpl_3} and \eqref{eq:ss_simpl_4}, we can write the dynamics of the overall mG as
\begin{subequations}
	\label{eq:1st_order_dyn}
	\begin{empheq}[left = \empheqlbrace]{align}
	\label{eq:1st_order_dyn_A}\mathbf{\dot{\Delta V}}  &= -\mathbb{Q} \mathbf{V} - \mathbb{L} D\mathbf{I_{L}}\\
	\label{eq:1st_order_dyn_B}\mathbf{\dot V} &= \Omega\mathbf{\Delta V}-\Omega \mathbf{V} + \Omega \mathbf{V_{ref}},
	\end{empheq}
\end{subequations}
or, equivalently, in compact form,
\begin{equation*}
\label{eq:compact_1st_order}
\quad
\begin{bmatrix}
\mathbf{\dot { {\Delta V}}}\\
\mathbf{\dot{V}}
\end{bmatrix} =\underbrace{ \begin{bmatrix}
	\mathbf{0}_{N\times N} & -\mathbb{Q}\\
	\Omega & -\Omega 
	\end{bmatrix} }_{\mathcal{Q}} \begin{bmatrix}
\mathbf{\Delta V}\\
\mathbf{V}
\end{bmatrix}+ \begin{bmatrix}
\mathbf{0}_{N\times N}& -\mathbb{L}D\\
\Omega & -\Omega 
\end{bmatrix}  \begin{bmatrix}
\mathbf{{V_{ref}}}\\
\mathbf{I_{L}}
\end{bmatrix},
\end{equation*}
with $\mathcal{Q}\in\mathbb{R}^{2N\times 2N}$ and $[\mathbf{\Delta V}^T \mathbf{V}^T]^T\in \mathbb{R}^N\times\mathbb{R}^N$.

\subsubsection{Analysis of equilibria}
The equilibria of
system \eqref{eq:1st_order_dyn} for constant
inputs $(\mathbf{I^*_L},\mathbf{V^*_{ref}})$, are obtained by
computing the solutions
$[\mathbf{\Delta V^*}^T,\mathbf{V^*}^T]^T$ to the system 
\begin{subequations}
	\label{eq:1st_order_eq}
	\begin{empheq}[left =  \empheqlbrace]{align}
	\label{eq:1st_order_eq_A}
	\mathbb{Q} \mathbf{V}^* &= - \mathbb{L}D \mathbf{I_{L}^*}\\
	\label{eq:1st_order_eq_B}
	\mathbf{0}_N &= \Omega\mathbf{\Delta V^*}-\Omega \mathbf{V^*} + \Omega \mathbf{V_{ref}^*}.
	\end{empheq}
\end{subequations}
Since matrix $\Omega$ is invertible, equation \eqref{eq:1st_order_eq_B} becomes
\begin{equation}
\label{eq:eq_1st_order_1}
\mathbf{V^*} = \mathbf{\Delta V^*} +\mathbf{V_{ref}^*}.
\end{equation}
By substituting \eqref{eq:eq_1st_order_1} in \eqref{eq:1st_order_eq_A}, we get
\begin{equation*}
\mathbb{Q} \mathbf{\Delta V^*} = - \mathbb{L}D \mathbf{I_{L}^*}-\mathbb{Q}\mathbf{V_{ref}^*}.
\end{equation*}
that is exactly \eqref{eq:eq_simpl}. We can then exploit Proposition \ref{prop:equilibrium} for concluding that there are infinitely many solutions $\mathbf{\Delta V^*\in\mathbb{R}}^N$ in the form \eqref{eq:infinite_equilibria}.
Replacing \eqref{eq:infinite_equilibria} in \eqref{eq:eq_1st_order_1}, we can write equilibria of system \eqref{eq:1st_order_dyn} as
\begin{equation}
\label{eq:equilibria_1st_order_1}
\left[ \begin{array}{c}
\mathbf{\Delta V^*}  \\
\mathbf{ V^*} 
\end{array}\right] = \left[ \begin{array}{c}
\mathbf{\widehat{\Delta V}^{*}} + \alpha\mathbf{1}_N\\
\mathbf{\widehat{\Delta V}^{*}} + \alpha\mathbf{1}_N+\mathbf{V_{ref}^*}
\end{array}\right].
\end{equation}
Relations between the equilibria and current sharing/voltage balancing are given in the next Proposition.
\begin{prp}
	\label{prop:volt_bal_2}
	Consider system \eqref{eq:1st_order_dyn} with constant inputs
	$\mathbf{I_{L}^*} $ and $\mathbf{ V_{ref}^*}$. At the equilibrium,
	current sharing is achieved. Moreover, if $\mathbf{V_{ref}^*} =
	V_{ref}\mathbf{1}_N$ (i.e. if Assumption \ref{ass:vref} holds) and
	$\alpha$ in \eqref{eq:equilibria_1st_order_1} is equal to zero, then
	the equilibrium $[\mathbf{\Delta V^*}^T,\mathbf{V^*}^T]^T$ verifies the voltage balancing.
\end{prp}
\begin{proof}
	Since equation \eqref{eq:ss_simpl_1} holds, one has that, at the
	equilibrium, relation \eqref{eq:I_curr_sharing} is verified. Then, the
	proof is  identical to the one of Proposition \ref{prop:volt_bal}.
\end{proof}
Next, similarly to the simplified case described in Section \ref{sec:st_an_1}, we evaluate the stability properties of the closed-loop system \eqref{eq:1st_order_dyn} so as to show the convergence of state $[\mathbf{\Delta V}^T \mathbf{V}^T]^T$ to an equilibrium which guarantees current sharing and voltage balancing.
\subsubsection{Stability analysis}
\begin{prp}
	\label{prop:7}
	If $[\mathbf{\Delta V}^T \mathbf{V}^T]^T$ verifies \eqref{eq:1st_order_dyn}, then $$
	\underbrace{\left[ \begin{array}{c}
		\mathbf{\Delta V}  \\
		\mathbf{ V} 
		\end{array}\right]}_{\mathbf{v}} = \underbrace{\left[ \begin{array}{c}
		\mathbf{\overline{\Delta V}}  \\
		\mathbf{\overline V} 
		\end{array}\right]}_{\mathbf{\bar v}} + \underbrace{\left[ \begin{array}{c}
		\mathbf{\widehat{\Delta V}}  \\
		\mathbf{\widehat V} 
		\end{array}\right]}_{\mathbf{\hat v}},
	$$ where $\mathbf{\bar v}\in H^1_{\perp}\times H^1_{\perp}$ and $\mathbf{\hat v}\in H^1\times H^1$ fulfill
	\begin{subequations}
		\label{eq:subsys_perp}            
		\begin{empheq}[left =  {\widetilde\Sigma}_{\perp}^1:\empheqlbrace]{align}
		\label{eq:sysdist_perp_A}\mathbf{\dot{\overline{\Delta V}}} &= \mathbf{0}_N\\
		\label{eq:sysdist_perp_B} \mathbf{\dot{\overline{V}}} &= \Omega\mathbf{\overline {\Delta V}}-\Omega \mathbf{\overline{V}} + \Omega\mathbf{\overline{V_{ref}}}                          
		\end{empheq}
	\end{subequations}
	and
	\begin{subequations}
		\label{eq:subsys_H1}            
		\begin{empheq}[left = {\widetilde\Sigma}^1:\empheqlbrace]{align}
		\label{eq:subsys_H1_A}\mathbf{\dot{\widehat{\Delta V}}} &= -\mathbb{Q}\mathbf{\widehat V} -\mathbb{L}D\mathbf{{I_L}}\\
		\label{eq:subsys_H1_B}\mathbf{\dot {\widehat V}} &= \Omega\mathbf{\widehat {\Delta V}}-\Omega \mathbf{\widehat V} + \Omega \mathbf{\widehat{V_{ref}}},
		\end{empheq}
	\end{subequations}
	respectively.
\end{prp}
\begin{proof}
	The dynamics of $\mathbf{\overline{\Delta V}}$ and $\mathbf{\widehat{\Delta
			V}}$ can be derived proceeding as in the proof of Proposition \ref{pr:invariant}. 
	In a similar way, by averaging both sides of
	\eqref{eq:1st_order_dyn_B}, one derives the (independent) dynamics of $\mathbf{\overline{V}}$ and $\mathbf{\widehat{V}}$. 
\end{proof}
The above decomposition allows us to evaluate the evolution of state $\mathbf{v}$ on $\mathbb{R}^{N\times N}$ by separately analyzing dynamics \eqref{eq:subsys_perp} and \eqref{eq:subsys_H1}, i.e. studying the behavior of projections $\mathbf{\bar v} = P_{\mathcal{V}}(\mathbf{v})$ and $\mathbf{\hat v} = P_{\mathcal{W}}(\mathbf{v})$, with $\mathcal{V} = H^1_{\perp}\times H^1_{\perp}$ and $\mathcal{W}= H^1\times H^1$. \\
First we focus on $\tilde{\Sigma}_{\perp}^1$. Without loss of
generality, for stability analysis we can neglect the input vector $\mathbf{\overline{V_{ref}}}$ in \eqref{eq:sysdist_perp_B}, thus having:
\begin{subequations}
	\label{eq:subsys_perp_noinputs}            
	\begin{empheq}[left =  \empheqlbrace]{align}
	\label{eq:sysdist_perp_A_free}\mathbf{\dot{\overline{\Delta V}}} &= \mathbf{0}_{N}\\
	\label{eq:sysdist_perp_B_free} \mathbf{\dot{\overline{V}}} &= \Omega\mathbf{\overline {\Delta V}}-\Omega \mathbf{\overline{V}}     .                    
	\end{empheq}
\end{subequations}
By construction, \eqref{eq:sysdist_perp_B_free} collects the decoupled equations
\begin{equation}
\label{eq:single_Vtilde}
\dot {\overline V}_i = \omega_c \overline{\Delta V_{i}}-\omega_c \overline V_{i}\hspace{5mm}\forall i = 1,\dots, N,
\end{equation}
where, according to \eqref{eq:sysdist_perp_A}, each term $\overline{\Delta V_{i}}$ in can be treated as an exogenous input (thus not affecting stability properties). It follows that dynamics \eqref{eq:single_Vtilde} is asymptotically stable, since $\omega_c>0$. In summary, system \eqref{eq:subsys_perp} tells us that the average $\mathbf{{\overline{\Delta V}}}$ will remain constant in time (and equal to $\langle\mathbf{{\Delta V_0}}\rangle$), while $\mathbf{\overline V}$ will converge to the origin. For studying stability properties of system $\widetilde{\Sigma}^1$, we consider \eqref{eq:1st_order_dyn} without inputs, i.e. 
\begin{equation}
\label{eq:subsys_H1_free}            
\left\lbrace
\begin{aligned}
\mathbf{\dot{{ \Delta V}}} &= -{ \mathbb{Q}}\mathbf{{ V}}\\
\mathbf{{ \dot V}} &=\Omega \mathbf{\Delta V}-\Omega\mathbf{V}
\end{aligned}
\right.
\end{equation}
and analyze stability on $H^1\times H^1$. We have the following result.
\begin{theorem}
	\label{thm:stability_1st_order}
	Under Assumption \ref{ass:D_identity} or \ref{ass:same_topology}, the origin of \eqref{eq:subsys_H1_free} is GES on $H^1\times H^1$. Furthermore, matrix $\mathcal{Q}$ has a simple zero eigenvalue and the rate of convergence is the maximum among real parts of all other eigenvalues.
\end{theorem}
\begin{proof}
	The proof is given in Appendix \ref{appendix_2}.
\end{proof}
By means of Theorem \ref{thm:stability_1st_order}, we have that, given
an initial condition $[\mathbf{\Delta V_0}^T \mathbf{V_0}^T]^T$ for system
\eqref{eq:1st_order_dyn}, the state $[\mathbf{\Delta V}^T \mathbf{V}^T]^T$
will converge to the equilibrium in \eqref{eq:equilibria_1st_order_1},
with $\alpha = \mathbf{\langle\Delta V_0\rangle}$.

The results above show that, for system \eqref{eq:1st_order_dyn}, current sharing is achieved in a GES fashion. In a similar way,  asymptotic voltage
balancing is ensured if Assumption
\ref{ass:vref} and \eqref{eq:cs_vb} are fulfilled.

\subsection{PnP design of secondary control}
\label{sec:pnp_sec_ctrl}
We now describe the procedure for designing secondary controllers in a PnP fashion. We will show that, as for the PnP design of primary regulators, when a DGU is added or removed, the secondary control layer can be updated only locally for preserving current sharing and voltage balancing.

\paragraph{Plug-in operation.} Under Assumption \ref{ass:D_identity}, when a DGU (say DGU $i$) sends a plug-in request at a time $\bar t$, it choses a
set $\mathcal{N}_{i}^c$ of communication neighbors (which must not necessarily coincide with $\mathcal{N}_{i}^{el}$) and fixes parameters $a_{ij}>0$, $\forall
j\in \mathcal{N}_{i}^c$, in order to design the local consensus filter
\eqref{eq:basic_consensus}. At the same time, each DGU $j$, $j\in\mathcal{N}_i^c$, updates its consensus filter by setting $a_{ji} = a_{ij}$ in \eqref{eq:basic_consensus}. If, instead, $D\neq I$, one can fulfill Assumption \ref{ass:same_topology} by allowing the entering DGU to receive the value of conductances $\frac{1}{R_{ij}}$ from neighboring DGUs $j\in \mathcal{N}_i^{el}$ and by setting $\mathcal{N}_i^{c} =\mathcal{N}_i^{el}$  (see Remark \ref{rmk:same_topology}). By doing so, DGU $i$ can choose parameters $a_{ij} = \mu\frac{1}{R_{ij}}$, and each DGU $j\in\mathcal{N}_i^{c}$ sets $a_{ji} =a_{ij}$, thus updating its consensus filter \eqref{eq:basic_consensus}.

Overall, Theorems \ref{thm:GES_H1} and
\ref{thm:stability_1st_order} ensure that the disagreement dynamics of
the mG states is GES.
Let Assumption \ref{ass:vref}
hold for all the interconnected DGUs in the mG before $\bar t$ and let
us denote the common reference voltage by $V_{ref}$. If DGU $i$ sets
$V_{ref,i} = V_{ref}$ and if we choose $\Delta V_i(\bar t) = 0$ (thus having $\langle\left[ \begin{array}{c}
\mathbf{\Delta V'}(\bar t) \\
\Delta V_i (\bar t)
\end{array}\right]\rangle= 0$, where $\mathbf{\Delta V'}(\bar t)$ is the
vector $\mathbf{\Delta V}$ prior the plugging-in of DGU $i$), both current sharing and
voltage balancing are preserved in the asymptotic r\'{e}gime (see Propositions
\ref{prop:volt_bal} and \ref{prop:volt_bal_2}). 

\paragraph{Unplugging operation.} Under Assumption \ref{ass:D_identity} or \ref{ass:same_topology}, when a DGU (say DGU $j$) is
unplugged at time $\bar t$, provided that the new graphs $\mathcal{G}_{el}$ and $\mathcal{G}_c$ fulfill Assumption \ref{ass:conneced}, the key condition that must
be guaranteed is that the vector $\mathbf{\Delta V}_{-j}$
(i.e. $\mathbf{\Delta V}$ without element $j$) verifies
$\mathbf{\langle\Delta V}_{-j}(\bar t)\rangle=0$. If $\mathbf{\langle\Delta
	V}(\bar t^-)\rangle=0$, this can be achieved by re-setting $\Delta
V_i(\bar t)=\Delta V_i(\bar t^-)+\frac{\Delta V_j(\bar
	t^-)}{|\mathcal{N}_j^c|}$ for all $i\in\mathcal{N}_j^c$, and keeping $\Delta V_i(\bar t)=\Delta V_i(\bar
t^-)$ for all $i \notin\mathcal{N}_{j}^c\cup \{j\}$. Indeed, this yields

\begin{equation*}
\begin{aligned}
&\mathbf{\langle\Delta V_{-j}}(\bar t)\rangle
=\frac{1}{N-1}\sum\limits_{i\in\mathcal{N}_j^c}\Delta V_i(\bar
t)+\frac{1}{N-1}\sum\limits_{i\notin\mathcal{N}_j^c\cup\{j\}}\Delta V_i(\bar
t)=\\
&=\frac{1}{N-1}\hspace{-1mm}\left(\sum\limits_{i\in\mathcal{N}_{j}^c}\hspace{-1mm}\Delta V_i(\bar
t^-) +\frac{1}{|\mathcal{N}_{j}^c|}\sum\limits_{i\in\mathcal{N}_{j}^c}\hspace{-1mm}\Delta V_j(\bar
t^-)+ \hspace{-2mm}\sum\limits_{i\notin\mathcal{N}_j^c\cup\{j\}}\hspace{-4mm}\Delta V_i(\bar t^-) \right)=\\
&=\frac{1}{N-1}\left(N \mathbf{\langle\Delta V}(\bar t^-)\rangle\right)=0.
\end{aligned}
\end{equation*}

\section{Validation of secondary controllers}
\subsection{Simulation results}
\label{sec:simulations}
In this Section, we demonstrate the capability of the proposed control
scheme to guarantee current sharing and voltage balancing when DGUs are plugged in/out and load changes occur. Simulations have been performed in Simulink/PLECS \cite{allmeling2013plecs}.
We consider an mG composed of 7 DGUs, interconnected as in Figure
	\ref{fig:sim_stages}, with non-identical electrical parameters and power lines. Notice the presence of loops in the electrical network, that complicates voltage regulation. Primary PnP voltage regulators are designed according to the method in \cite{tucci2015decentralized}, whereas, for the secondary control layer, we choose $k_I$ in \eqref{eq:ss_simpl_1} equal to 1. DGUs have rated currents ${I}_{ti}^r = 10$ A, $i = 1,2,3$, ${I}_{t4}^r = {I}_{t5}^r = 5$ A and ${I}_{t6}^r = {I}_{t7}^r = 3.33$ A. Since, in this scenario, we aim to achieve the current sharing condition \eqref{eq:cs_defn}, we set $I_{ti}^s = I_{ti}^r$, $i= 1,\dots,7$, thus having, in \eqref{eq:ss_simpl_1}, $D\neq I$. Then, in order to guarantee the asymptotic stability of the hierarchical control scheme, we fulfill Assumption \ref{ass:same_topology} by (\textit{i}) letting $\mathcal{G}_{c}$ have the same topology of $\mathcal{G}_{el}$, and (\textit{ii}) picking $a_{ij} = \frac{1}{R_{ij}}$, if DGUs $i$ and $j$ are connected by a communication link (see Remark \ref{rmk:same_topology}). Furthermore, the voltage reference in Assumption \ref{ass:vref} is $V_{ref} = 48$ V, and the electrical parameters are given in Appendix \ref{sec:AppElectrPar}.
 
In the following, we describe Figure
\ref{fig:sim_complete_1}, which illustrates the evolution of the main
electrical quantities (i.e. measured DGU output currents in Amperes, DGU output currents in p.u., PCC voltages and
average PCCs voltage) during the consecutive simulation
stages shown in Figure~\ref{fig:sim_stages}.

\paragraph{Stage 1:}
At time $t_0 = 0$, all the DGUs are isolated and only
the primary PnP voltage regulators, designed as in \cite{tucci2015decentralized}, are active. Therefore, as shown in
Figure \ref{fig:sim_complete_1}, (\textit{i}) each DGU supplies its local load while keeping the
corresponding PCC voltage at 48 V, and (\textit{ii}) the DGU output currents in p.u.
are different. We further highlight that primary controllers have been designed assuming that all the switches in Figure \ref{fig:sim_stages} connecting DGUs 1-6 are closed. From \cite{tucci2015decentralized}, however, they also stabilize the mG when all switches are open.

\paragraph{Stage 2:}
Subsystems 1-6 are
	connected together at time $t_1 = 2$ s and, as described before, no update of primary controllers is needed. As shown in the plot of $V_{PCC}$ in Figure \ref{fig:sim_complete_1}, voltage stability and fast transients after the plug-in operations are ensured by PnP primary regulators. The secondary control layer is still disabled at this stage.

\paragraph{Stage 3:}
At time $t_2 = 5$ s, we activate the secondary control layer for DGUs 1-6, thus
ensuring asymptotic current sharing among them (see the plot of the currents in p.u.). This is achieved by automatically
adjusting the voltages at PCC (as shown in the plot of $V_{PCC}$).Moreover, the top plot of Figure \ref{fig:sim_complete_1} reveals that, as expected, DGUs 4 and 5 share half of the current of DGUs 1-3, while the output current of DGUs 6 is one third of the currents of DGUs 1-3. We further highlight that, by setting $\Delta V_i(t_{2}) = 0$, $i =
	1,\dots,6$, as described in Section \ref{sec:pnp_sec_ctrl}, condition
	\eqref{eq:cs_vb} is fulfilled and asymptotic voltage balancing is guaranteed (see the plot of the average PCCs voltage, indicated with $V_{av}$).
	
\paragraph{Stage 4:} For evaluating the PnP capabilities of our control scheme, at $t_3
=15$ s, DGU 7 sends a plug-in request to DGUs
4 and 5. Previous primary controllers of DGUs 4 and 5
still fulfill the plug-in conditions in \cite{tucci2015decentralized}: they are therefore maintained and the plug-in of DGU 7 is performed. In the light of Assumption \ref{rmk:same_topology}, at $t_3$ also the secondary controller of DGU 7 is activated, thus enabling the DGU to 
	contribute to current sharing. This can be noticed in Figure
\ref{fig:sim_complete_1}, as all PCC voltages change in order
to let all the output currents in p.u. converge to a common
value. We also notice that the measured output currents (top plot of Figure \ref{fig:sim_complete_1}) are still shared accordingly (i.e. $I_{t1} = I_{t2} = I_{t3} = 2I_{t4} = 2I_{t5} = 3I_{t6}= 3I_{t7}$). Furthermore, choosing $\Delta V_7(t_3) = 0$ (as described in
Section \ref{sec:pnp_sec_ctrl}), we maintain the average PCCs voltage at 48 V (see $V_{av}$, stage 4). 

\paragraph{Stage 5:} At $t_4 = 25$ s, we halve the load of DGU 1,
thus increasing the corresponding load current $I_{L1}$ and causing a
peak in the corresponding output current. However, after few seconds,
all the DGUs share again the total load current, while the averaged PCC voltage converges to the reference value. 

\paragraph{Stage 6:}  Finally, we assess the performance of the proposed hierarchical scheme when the sudden disconnection of a DGU occurs. To this aim, at time $t_5 = 35$ s, we disconnect DGU 3. Figure \ref{fig:sim_complete_1} (stage 6) shows that voltage stability, current sharing and voltage balancing are preserved in the mG composed of DGU 1-2 and 4-7.

\begin{figure}
	\centering
	\includegraphics[scale=.3]{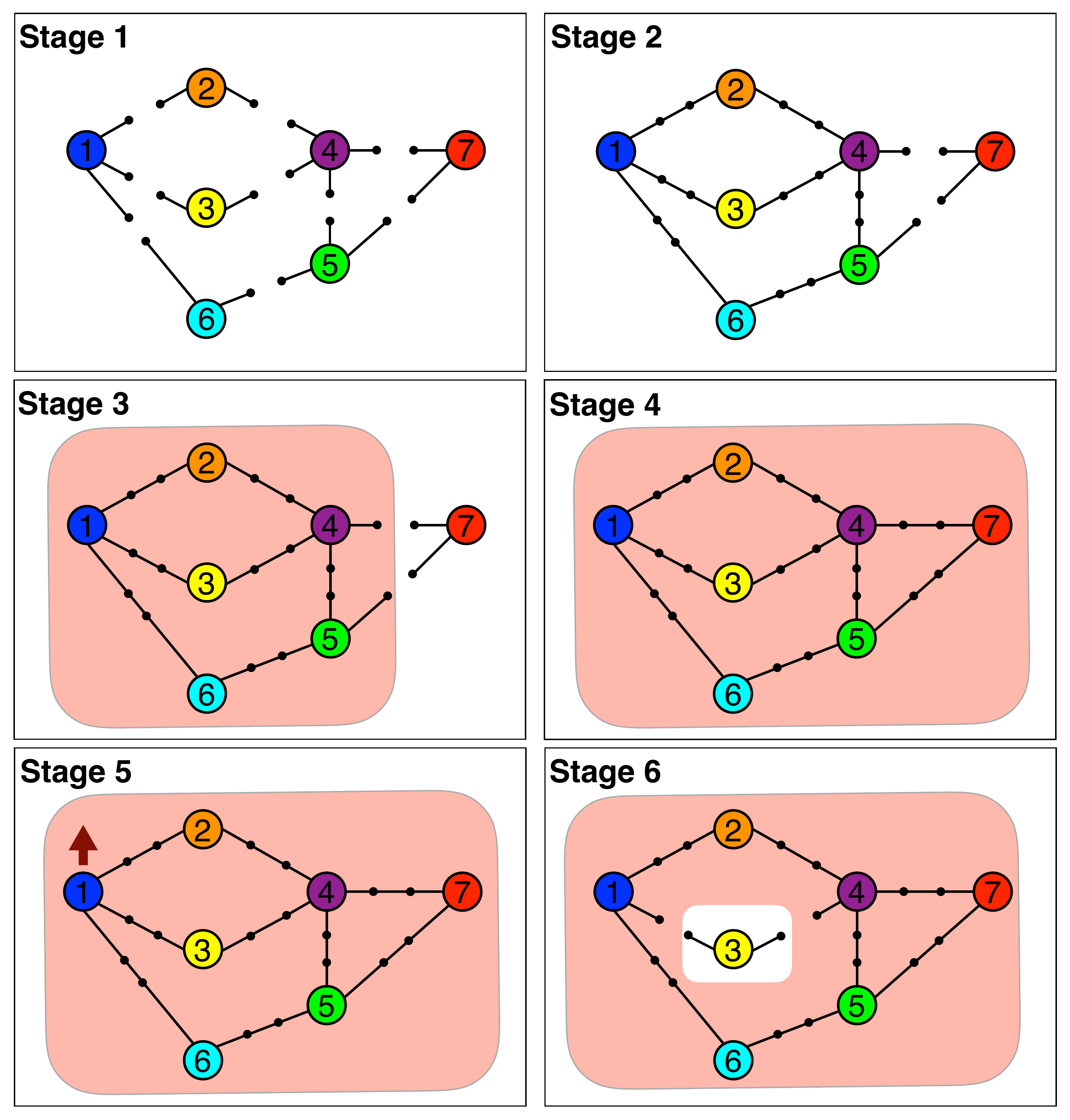}
	\caption{Simulation stages: numbered nodes represent DGUs, while black lines denote power lines. The secondary control layer is activated for the DGUs contained in the red area. Open switches in stages 1, 2, 3 and 6 denote disconnected DGUs. The arrow in stage 5 represents a step up in the load current of DGU 1.}
	\label{fig:sim_stages}
\end{figure}
\begin{figure}
	\centering
	\includegraphics[scale=0.4]{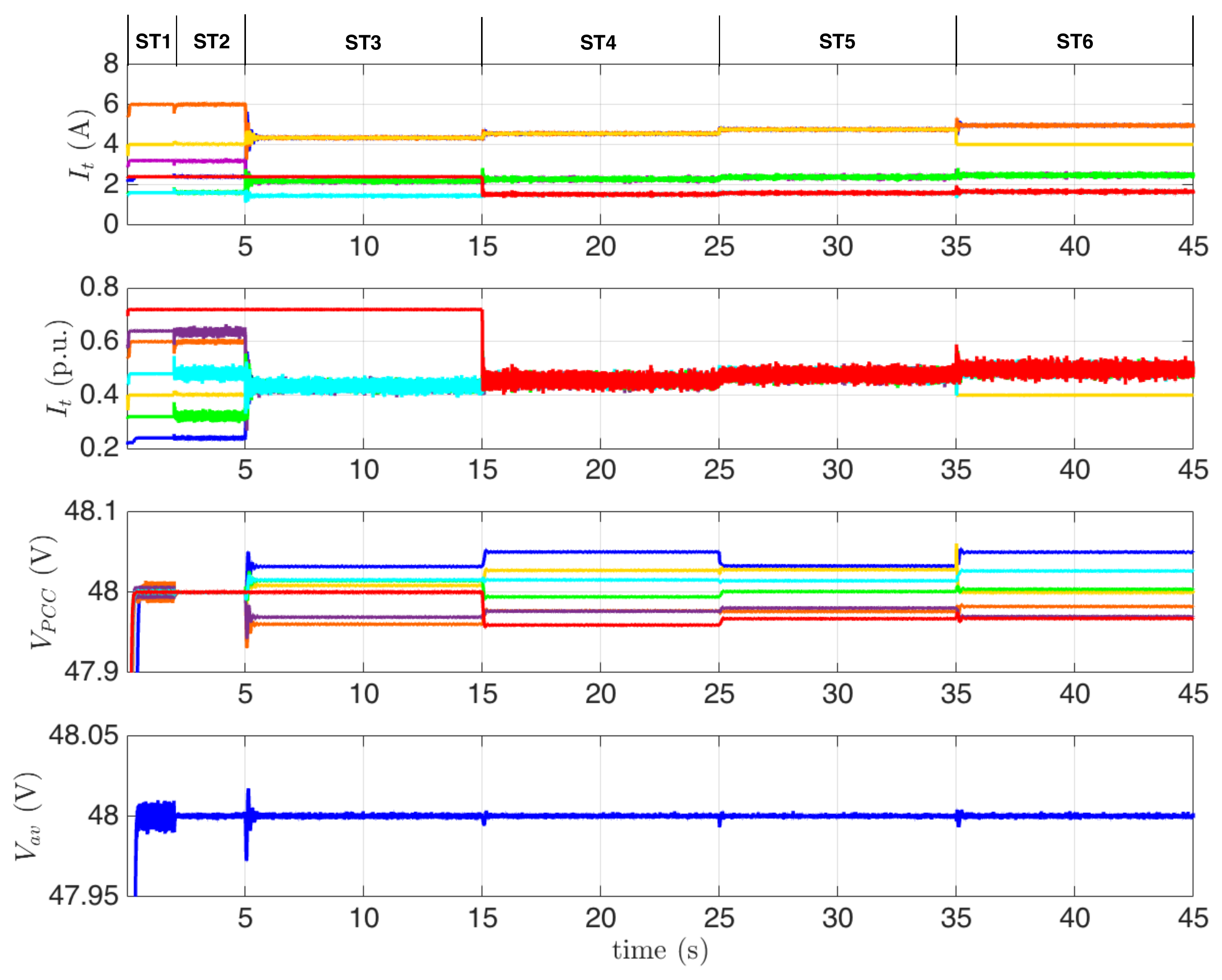}
	\caption{Simulation results: evolution of measured output currents, output currents in p.u., 
		voltages at PCCs and average PCCs voltage. Lines in the
		plots are associated with different DGUs and they are color-coded as in Figure \ref{fig:sim_stages}. Simulation stages are those shown in Figure \ref{fig:sim_stages}.}
	\label{fig:sim_complete_1}
\end{figure}
\subsection{Experimental results}
\label{sec:experiments}
Performance brought about by the presented hierarchical scheme been also validated via experimental tests based on the mG platform in the top panel of Figure \ref{fig:experimental_setup}, which consists of three \textit{Danfoss} inverters, a dSPACE1103 control board and LEM sensors. In order to properly emulate DC/DC converters (i.e. Buck converters), only the first phase of each inverter has been used. Buck converters operate in parallel to emulate DGUs while different local load conditions have been obtained by connecting each PCC to a resistive load. All the converters are supplied by DC source generators.

For this scenario, primary PnP voltage controllers are designed using the approach in \cite{tucci2016improved}, and Assumption \ref{ass:D_identity} holds. Hence, since $D = I$, $I_{t1}^s = I_{t2}^s = I_{t3}^s = \bar{I}_t$, i.e. we aim to achieve the asymptotic current sharing condition \eqref{eqn:equal_sharing}. Moreover, we set $k_I$ in \eqref{eq:ss_simpl_1} equal to 0.5, while coefficients $a_{ij}$ in \eqref{eq:basic_consensus} are equal to 1 if DGUs $i$ and $j$ are connected by a communication link, 0 otherwise. We also recall that, under Assumption \ref{ass:D_identity}, the stability of the mG equipped with our hierarchical scheme is preserved even if the topologies of $\mathcal{G}_c$ and $\mathcal{G}_{el}$ differ. Notably, we consider the mG in Figure \ref{fig:experimental_setup}, where $\mathcal{G}_{el}$ and $\mathcal{G}_{c}$ are highlighted in blue and red, respectively, and the edges of $\mathcal{G}_{el}$ are $RL$ lines.
	
The controllers have been implemented in Simulink and compiled to the dSPACE system in order to command the Buck switches at a frequency of 10 kHz.

The evolution of the main electrical quantities is shown in Figure \ref{fig:experiments}. At time $t_0 = 0$ s, all the DGUs are isolated and do not communicate. At times $t_1\approx 2.5$ s, $t_2\approx 5$ s and $t_3\approx 10$ s, we connect DGU 1 to 2, 2 to 3 and 1 to 3, respectively, thus obtaining a loop in the electrical topology. As described in \cite{tucci2016improved}, no update of primary PnP controllers is required when units are connected. As shown in the plot of the PCC voltages in Figure \ref{fig:experiments}, PnP primary voltage regulators ensure smooth transitions and stability. Since for $t < t_3$ the secondary layer is not active, the output currents are not equally shared and the PCC voltages coincide with the reference (notably, Assumption \ref{ass:vref} holds, with $V_{ref} = 48$ V). Next, at time $t_4\approx 15$ s, we set $\Delta V_i(t_{4}) = 0$, $i =
1,2,3$ and enable the secondary control layer. As expected, the three output currents converge to the same value (see $I_{t}$ in Figure \ref{fig:experiments}). Furthermore, the fulfillment of condition \eqref{eq:cs_vb} guarantees asymptotic voltage balancing (see $V_{av}$ in Figure \ref{fig:experiments}). Finally, at time $t_5\approx 25$ s we decrease the load of DGU 3, causing an increment in the corresponding load current. As a consequence, the value of $\bar {I}_{t} = \langle\mathbf{{I}_L^*}\rangle$ increases as well. Also in this case, the total load current is equally shared among DGUs while $V_{av}$ does not deviate from 48 V.

As for the voltage discrepancies induced by secondary controllers, they are within 5$\%$ of the nominal voltage $V_{ref} = 48$ V. This fulfills, for instance, the standards for DC mGs used as uninterruptible power supply systems for telecommunication applications \cite{schulz2007etsi}.

\begin{figure}
	\centering
	\includegraphics[scale=0.65]{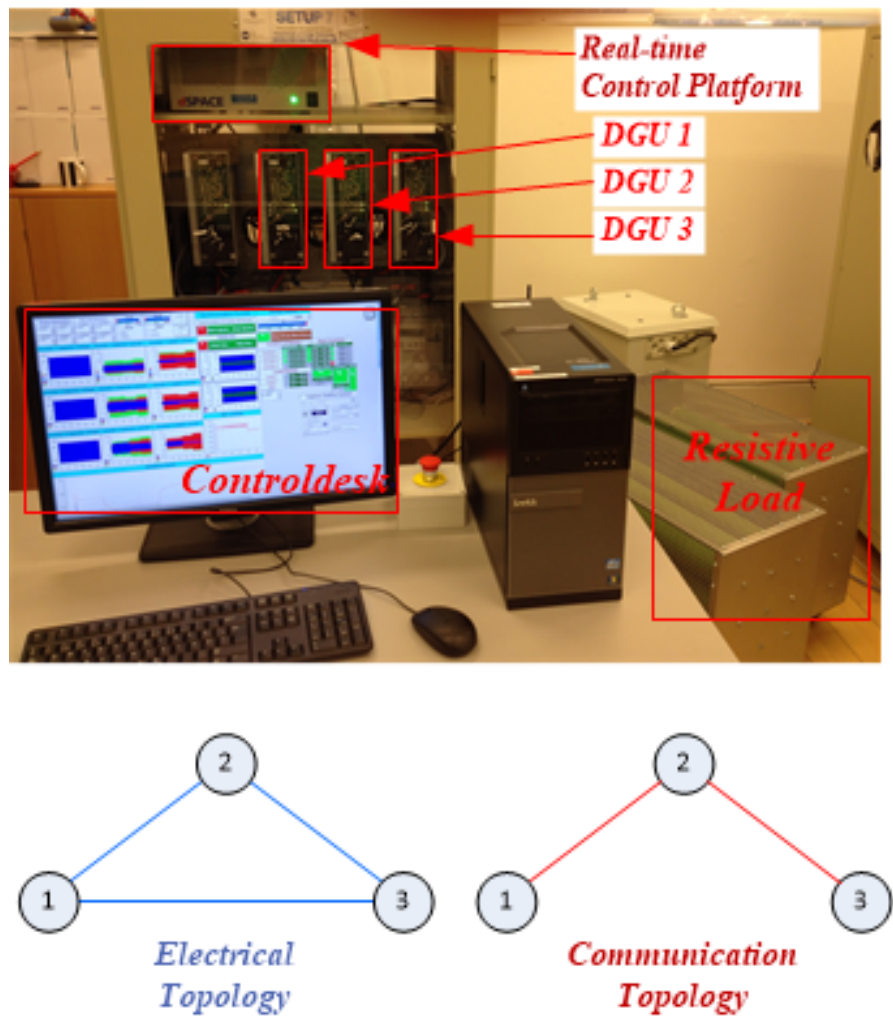}
	\caption{Experimental validation: mG setup (top) and topologies of the electrical and communication graphs (bottom).}
	\label{fig:experimental_setup}
\end{figure}
\begin{figure}
	\centering
	\hspace{-4mm}
	\includegraphics[scale=0.4]{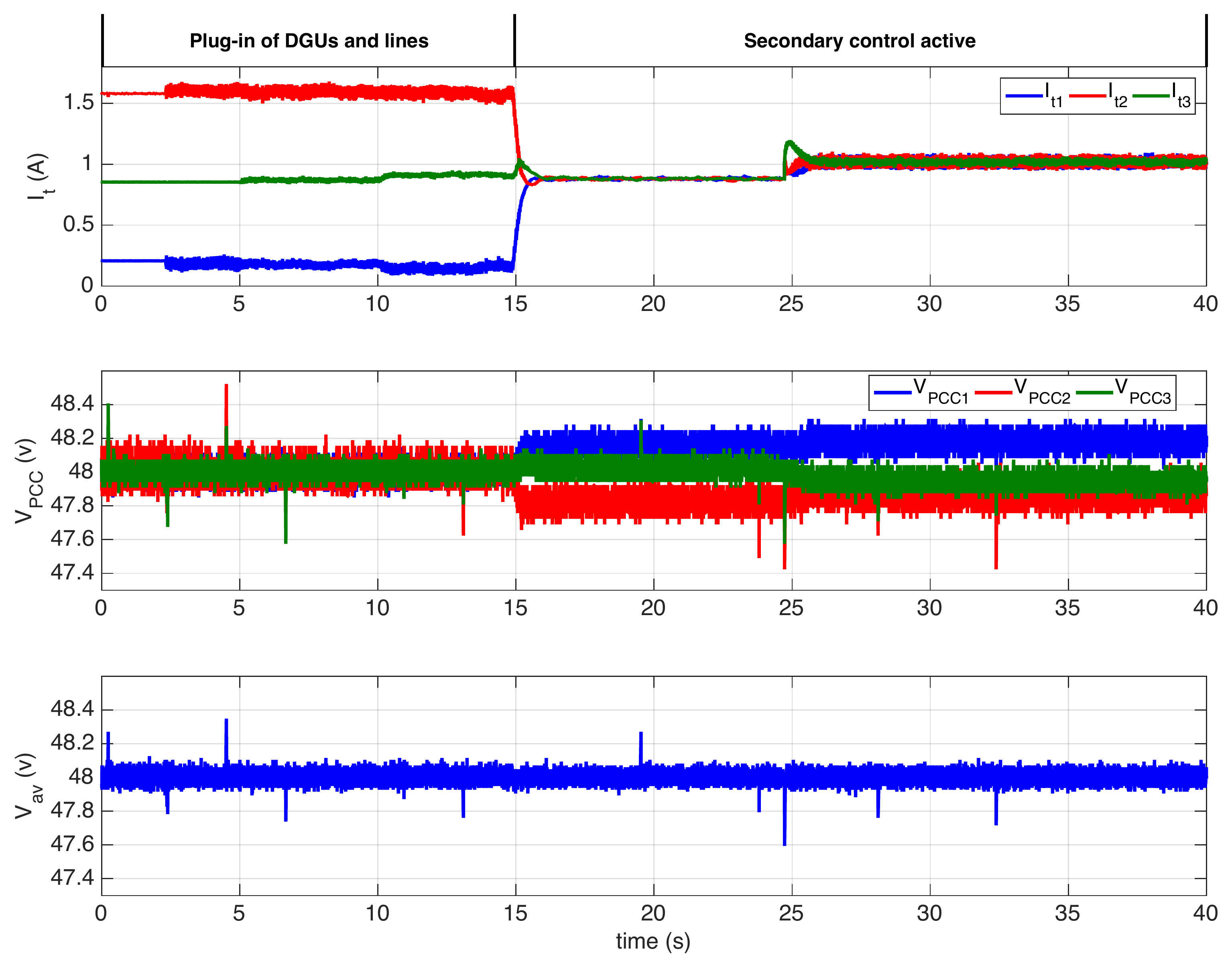}
	\caption{Experimental results for the mG in Figure \ref{fig:experimental_setup}. In the time interval from $2.5$ s to $15$ s, DGUs are connected together and primary PnP voltage regulators are enabled. From time $15$ s onwards, also the secondary control layer is active. At time $25$ s, the local load of DGU 3 is halved.}
	\label{fig:experiments}
\end{figure}

\section{Conclusions}
In this paper, a secondary consensus-based control layer for
current sharing and voltage balancing in DC mGs has been
presented. Under the assumption that DGUs are equipped with decentralized primary controllers that guarantee voltage stability, we proved stability of the hierarchical control scheme. Moreover, we presented a method for designing secondary
controllers in a PnP fashion for handling plugging -in/-out of DGUs. 

Future developments will study the impact of non-idealities (such as transmission delays, data quantization and packet drops) on the performance of closed-loop mGs. The analysis of networked control systems has received considerable attention in the recent past \cite{Hespanha2007} and our goal will be to reappraise methods and tools developed within this area in the context of microgrids. Another fundamental challenge that must be addressed when higher-level networked schemes are included in the mG control architecture is security against cyberattacks \cite{Pasqualetti2013}.

     \clearpage

     \appendix
    \section{Proof of Theorem \ref{thm:GES_H1}}
\label{appendix_1}
We introduce a preliminary Lemma, partly taken from Theorem 19 in \cite{callier2012linear}.
\begin{lem}
	\label{lem:stability}
	For $\mathcal{A}\in\mathbb{R}^{n\times n}$, let $\mathcal{V}$ and $\mathcal{W}$ be $\mathcal{A}$-invariant subspaces of $\mathbb{R}^n$ such that dim$(\mathcal{V}) = k$, dim$(\mathcal{W}) = n-k$ and $\mathbb{R}^n = \mathcal{V}\oplus\mathcal{W}$. Then:
	\begin{enumerate}[I)]
		\item \label{thm:stability_1} there is a matrix $T\in\mathbb{R}^{n\times n}$ such that $A  = T^{-1}\mathcal{A} T$ has the block-diagonal structure 
		\begin{equation}
		\label{eq:repres_matrix}
		A = \left[ 
		\renewcommand\arraystretch{1.4}
		\begin{array}{c|cc}
		A_{11} &  \mathbf{0}  \\
		\hline
		\mathbf{0}  & A_{22}\\
		\end{array}\right],
		\end{equation}
		with $A_{11}\in\mathbb{R}^{k\times k}$ and $A_{22}\in\mathbb{R}^{(n-k)\times (n-k)}$.
		In particular, if $\{b_1, \dots, b_k\}$ and $\{b_{k+1},\dots,b_{n}\}$ are basis for $\mathcal{V}$ and $\mathcal{W}$, respectively, the transformation matrix $T$ has the block structure
		\begin{equation}
		\label{eq:matrix_T}
		T = [b_1 | \cdots |b_k|b_{k+1}|\cdots|b_{n}].
		\end{equation}
		Therefore, if $x\in\mathcal{V}$, then $T^{-1}x=[\tilde{x}_1^T \text{ } |	\text{ } 0]^T$, with $\tilde{x}_1 \in\Rset^k$. Similarly, if $x\in\mathcal{W}$, then $T^{-1}x=[ 0 \text{ } |	\text{ } \tilde{x}_2^T]^T$, with $\tilde{x}_2\in\Rset^{n-k}$.
		\item\label{thm:stability_2} The origin of $\dot x = \mathcal{A} x$ is GES on $\mathcal{V}$ if and only if the origin of  $\dot{\tilde{x}}_1 = A_{11}\tilde{x}_1$ is GES. Moreover, parameters $\kappa, \eta>0$ verifying $\lVert\tilde{x}_1(t) \rVert\leq\kappa e^{-\eta t}\lVert\tilde{x}_1(0) \rVert$, also guarantee $\lVert P_{\mathcal{V}}x(t) \rVert\leq\kappa e^{-\eta t}\lVert P_{\mathcal{V}}x(0) \rVert$.
	\end{enumerate}
\end{lem}
\begin{proof}
	For the proof of point \ref{thm:stability_1}, we defer the reader to the proof of Theorem 19 in \cite{callier2012linear}.
	
	The proof of point \ref{thm:stability_2} directly follows from the block-diagonal structure of matrix $A$ in \eqref{eq:repres_matrix}. Indeed,
	\begin{equation*}
	\dot{x} = \mathcal{A} x \Leftrightarrow \dot{\tilde x} =A \left[ \begin{array}{c}
	\tilde{x}_1 \\
	\hline
	\tilde{x}_2
	\end{array}\right]\Leftrightarrow
	\quad\left\lbrace \begin{aligned}
	\dot{\tilde{x}}_1  &= A_{11}\tilde{x}_1\\
	\dot{\tilde{x}}_2 &= A_{22}\tilde{x}_2,\\  
	\end{aligned}
	\right.
	\end{equation*}
	i.e. $A_{11}$ is the matrix representation of the map $\mathcal{A}(\mathcal{V}|\mathcal{V})$. In other words, studying the stability of $\mathcal{A}$ on $\mathcal{V}$ is equivalent to study the stability of $A_{11}$. 
	Moreover, by construction, $P_{\mathcal{V}}(x) = T[\tilde{x}_1^T \text{ } |	\text{ } 0]^T$. Then,
	\begin{equation}
	\label{eq:theroem_part2_1}
	\lVert P_{\mathcal{V}}(x(t))\rVert\leq \lVert T \rVert \lVert\tilde{x}_1(t) \rVert\leq \lVert T \rVert\kappa e^{-\eta t}\lVert\tilde{x}_1(0) \rVert.
	\end{equation}
	Since $\lVert\tilde{x}_1(0) \rVert \leq \lVert T^{-1}\rVert \lVert P_{\mathcal{V}}(x(0)) \rVert$, inequality \eqref{eq:theroem_part2_1} becomes
	\begin{equation*}
	\label{eq:theroem_part2_2}
	\lVert P_{\mathcal{V}}(x) \rVert\leq \lVert T \rVert\kappa e^{-\eta t}\lVert\tilde{x}_1(0) \rVert\leq \kappa e^{-\eta t}\lVert P_{\mathcal{V}}(x(0)) \rVert.
	\end{equation*}
\end{proof}
\paragraph{Proof of Theorem \ref{thm:GES_H1}.}
\label{app:GES_H1}
Points (\ref{prop:zero_row_sum}) and (\ref{prop:range}) of Proposition \ref{prop_2} show that subspaces $H^1$ and $H_{\perp}^1$ are $\mathbb{Q}$-invariant. Moreover, $\mathbb{R}^N = H^1 \oplus H_{\perp}^1$. It follows that Lemma \ref{lem:stability} can be applied with $\mathcal{V} = H^1$ and $\mathcal{W} = H^1_{\perp}$. In particular, by means of point \ref{thm:stability_1}, we know that there exists a transformation matrix $T\in\mathbb{R}^{N\times N}$ such that the linear map $\mathbb{Q}$ can be represented as in \eqref{eq:repres_matrix}. Denoting with $B^1 = \{b_1, \dots, b_{N-1}\}$ and $B_{\perp}^1=\mathbf{1}_N$ the basis for $H^1$ and $H_{\perp}^1$, respectively, from \eqref{eq:matrix_T}, we have
\begin{equation*}
T = [b_1 | \cdots |b_{N-1}|\mathbf{1}_N].
\end{equation*}
Matrix $Q = T^{-1}\mathbb{Q} T $ is given by 
\begin{equation}
\label{eq:ex_A}
\begin{aligned}
Q &=
\begin{bmatrix}
Q_{11} & \mathbf{0}\\ 
\mathbf{0} & q_{22} 
\end{bmatrix}
\end{aligned}
\end{equation}
where $Q_{11}\in\mathbb{R}^{(N-1)\times (N-1)}$. Moreover, scalar $q_{22} = 0$ since, by construction, it represents the map $\mathbb{Q}(H_{\perp}^1|H_{\perp}^1)$ (see Proposition \ref{prop_2}-(\ref{prop:zero_row_sum})). We notice that the representations of $\mathbf{\widehat{\Delta V}}$ and $\mathbf{\overline{\Delta V}}$ with respect to the basis $\mathcal{B}$ are  $\tilde v_1 = T^{-1}\mathbf{\widehat{\Delta V}} = [\xi_1 ,\dots, \xi_{N-1},0]^T$ and $\tilde v_{\perp} = T^{-1}\mathbf{\overline{\Delta V}} = [0 ,\dots, 0,\xi_N]^T$, respectively.
Now we prove that the origin of
\begin{equation}
\label{eq:origin_ideal}
\dot{\tilde{v}}_1 = -Q_{11}\tilde{v}_1
\end{equation}
is GES. Since $Q$ and $\Qset$ are similar matrices, they have the same
eigenvalues. Therefore, by exploiting points (\ref{prop:diag_and_psd}) and (\ref{prop:zero_eig}) of Proposition
\ref{prop_2}, one has that all the eigenvalues of $Q_{11}$
are strictly positive. This proves that \eqref{eq:origin_ideal} is GES
and, as shown in \cite{callier2012linear}, the convergence rate is $-\lambda$, where $\lambda$ is the minimal eigenvalue of $Q_{11}$.
The remainder of the proof follows directly from point \ref{thm:stability_2} of Lemma \ref{lem:stability}. \qed
    \newpage
    \section{Proof of Theorem \ref{thm:stability_1st_order}}
\label{appendix_2}
We first present two Propositions which provide preliminary results that will be used to prove Theorem \ref{thm:stability_1st_order}.
\begin{prp}
	\label{prop:h1_h1}
	Subspaces $H^1\times H^1$ and $H^1_{\perp}\times H^1_{\perp}$ are $\mathcal{Q}$-invariant.
\end{prp}
\begin{proof}
	We first show that, for any vector $\mathbf{\hat v} = [\mathbf{\hat{v}_1}^T \mathbf{\hat{v}_2}^T]^T$, it holds $\mathcal{Q}\mathbf{\hat v}\in H^1\times H^1$. Indeed,
	\begin{equation*}
	\qquad\mathcal{Q} \mathbf{\hat v} =
	\begin{bmatrix}
	\mathbf{0}_{N\times N} & -\Qset\\
	\Omega & -\Omega 
	\end{bmatrix}  \begin{bmatrix}
	\mathbf{\hat{v}_1}\\
	\mathbf{\hat{v}_2}
	\end{bmatrix}=\begin{bmatrix}
	-\mathbb{Q}\mathbf{\hat{v}_2}\\
	\Omega(\mathbf{\hat{v}_1}-\mathbf{\hat {v}_2})
	\end{bmatrix},
	\end{equation*}
	and, from Proposition \ref{prop_2}-(\ref{prop:Q_invertible}), the rightmost vector belongs to $H^1\times H^1$.\\
	Similarly, for any vector $\mathbf{\bar v} = [\mathbf{\bar{v}_1}^T \mathbf{\bar{v}_2}^T]^T\in H^1_{\perp}\times H^1_{\perp}$, we have that
	\begin{equation*}
	\qquad\mathcal{Q} \mathbf{\hat v} =
	\begin{bmatrix}
	\mathbf{0}_{N\times N}& -\Qset\\
	\Omega & -\Omega 
	\end{bmatrix}  \begin{bmatrix}
	\mathbf{\bar{v}_1}\\
	\mathbf{\bar{v}_2}
	\end{bmatrix}=\begin{bmatrix}
	\mathbf{0}_N\\
	\Omega(\mathbf{\bar{v}_1}-\mathbf{\bar {v}_2})
	\end{bmatrix}
	\end{equation*}
	and then $\mathcal{Q}\mathbf{\bar v}\in H_{\perp}^1\times H^1_{\perp}$.
\end{proof}
\begin{prp}
	\label{prop:eig_Qcall}
	Matrix $\mathcal{Q}$ has two eigenvalues equal to zero and $-\omega_c$, respectively. All other eigenvalues have strictly negative real part.
\end{prp}
\begin{proof}
	By definition, vector $[\mathbf{{\Delta V}}^T \mathbf{{V}}^T]^T\neq\mathbf{0}_{2N}$ is an eigenvector of $\mathcal{Q}$, if there exists $\lambda_i$ such that 
	\begin{equation}
	\label{eq:def_eig}
	\begin{bmatrix}
	\mathbf{0}_{N\times N} & -\Qset\\
	\Omega & -\Omega 
	\end{bmatrix}\begin{bmatrix}
	\mathbf{{\Delta V}}\\
	\mathbf{{V}}                 \end{bmatrix} = \lambda_i \begin{bmatrix}
	\mathbf{{\Delta V}}\\
	\mathbf{{V}}                 \end{bmatrix}.
	\end{equation}
	From \eqref{eq:def_eig}, one gets:
	\begin{subequations}
		\begin{align}
		\label{eq:def_eig_1A}
		-\mathbb{Q}\mathbf{{V}}&= \lambda_i\mathbf{{\Delta V}}\\
		\label{eq:def_eig_1B}
		\omega_c (\mathbf{{\Delta V}}-\mathbf{{V}}) &= \lambda_i {\mathbf{V}}.\end{align}
	\end{subequations}
	By isolating $\mathbf{{\Delta V}}$ in \eqref{eq:def_eig_1B} and substituting it in \eqref{eq:def_eig_1A}, we obtain
	\begin{equation}
	\label{eq:lambda_hat}
	-\mathbb{Q}{\mathbf{ V}} =\underbrace{ \frac{\lambda_i(\lambda_i+\omega_c)}{\omega_c}}_{\widehat \lambda_i}\mathbf{ V}.
	\end{equation}
	where $\hat\lambda_i$ are, by construction, eigenvalues of $-\mathbb{Q}$. From points (\ref{prop:diag_and_psd}) and (\ref{prop:zero_eig}) of Proposition \ref{prop_2}, we have 
	\begin{subequations}
		\begin{align}
		\label{eq:def_eig_hat_lambda_1}
		\hat\lambda_N &= 0\\
		\label{eq:def_eig_hat_lambda_2}
		\hat\lambda_i &=-\gamma_i,\hspace{2mm}\gamma_i>0\hspace{7mm}i = 1,\dots,N-1.\end{align}
	\end{subequations}
	From \eqref{eq:def_eig_hat_lambda_1} and \eqref{eq:lambda_hat}, one has
	\begin{equation*}
	\lambda_i(\lambda_i+\omega_c)=0
	\end{equation*}
	and hence $\mathcal{Q}$ has a single eigenvalue equal to zero and an eigenvalue equal to $-\omega_c$.
	By substituting in \eqref{eq:def_eig_hat_lambda_2} the expression of $\hat\lambda_i$ in \eqref{eq:lambda_hat}, one gets:
	\begin{equation}
	\label{eq:eig_hur}
	\frac{\lambda_i^2}{\omega_c} + \lambda_i + \gamma_i = 0\hspace{5mm}i =1,\dots,N-1.
	\end{equation}
	Since all the coefficients of the polynomial in \eqref{eq:eig_hur} are strictly positive, we can conclude that matrix $\mathcal{Q}$ has $2(N-1)$ eigenvalues with Re$(\lambda_i)<0$.
\end{proof}
\paragraph{Proof of Theorem \ref{thm:stability_1st_order}.}
Similarly to the proof of Theorem \ref{thm:GES_H1}, we can exploit Lemma \ref{lem:stability} with $\mathcal{V}= H^1\times H^1$ and $\mathcal{W} = H^1_{\perp}\times H^1_{\perp}$. In fact, we know that (\textit{i}) subspaces $H^1\times H^1$ and $H^1_{\perp}\times H^1_{\perp}$ are $\mathcal{Q}$-invariant (see Proposition \ref{prop:h1_h1}) and (\textit{ii}) $\Rset^{N}\times\Rset^N = \mathcal{V} \oplus\mathcal{W}$. Hence, there exists a transformation matrix $T\in\Rset^{2N\times 2N}$ such that the linear map $\mathcal{Q}$ has an equivalent block-diagonal representation of the form \eqref{eq:repres_matrix}, i.e. \vspace{3mm}
\begin{equation}
\label{eq:ex_B}
T^{-1}\mathcal{Q} T = Q =
\left[ 
\renewcommand\arraystretch{1.4}
\begin{array}{c|cc}
Q_{11} &  \mathbf{0}  \\
\hline
\mathbf{0}  & Q_{22}\\
\end{array}\right],
\end{equation}
with $Q_{11}\in\mathbb{R}^{2(N-1)\times 2(N-1)}$ and $Q_{22}\in\mathbb{R}^{2\times 2}$.
By construction, matrices $Q_{11}$ and $Q_{22}$ in \eqref{eq:ex_B}
represent the maps $\mathcal{Q}(\mathcal{V}|\mathcal{V})$ and $\mathcal{Q}(\mathcal{W}|\mathcal{W})$, respectively. In particular, in the light on the consideration made for system
\eqref{eq:subsys_perp}, we have that the eigenvalues of $Q_{22}$ are
zero and $-\omega_c$. Moreover, by construction, the eigenvalues of
$Q_{11}$ are the $2(N-1)$ eigenvalues of $\mathcal{Q}$ with strictly negative
real part (see Proposition \ref{prop:eig_Qcall}). \qed
    \newpage
    \section{On the eigenvalues of $\mathbb{Q} = \mathbb{L}D\mathbb{M}$}
\label{app:eig_Q}
In this appendix, we provide an example which shows that, by pre- and  post- multiplying a generic positive definite diagonal matrix by two positive semidefinite Laplacians (associated with graphs having different topologies), one can obtain a matrix with some negative eigenvalues. 

Let us consider the graphs depicted in Figures \ref{fig:G_1} and \ref{fig:G_2}. 
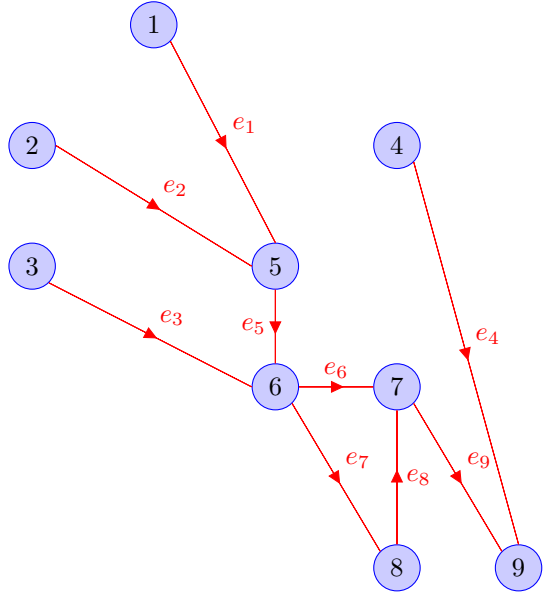
\begin{figure}[H]
	\centering
	\ctikzset{bipoles/length=1.2cm}
	\tikzstyle{every node}=[font=\normalsize]
	\begin{circuitikz}[american currents, scale=0.8]
		\draw (0,0) node(a) [circle, draw=blue, fill=blue!20] {$1$};
		\draw (-2,-2) node(b)  [circle, draw=blue, fill=blue!20] {$2$};
		\draw (-2,-4) node(c)  [circle, draw=blue, fill=blue!20] {$3$};
		\draw (4,-2) node(d)  [circle, draw=blue, fill=blue!20] {$4$};
		\draw (2,-4) node(e)  [circle, draw=blue, fill=blue!20] {$5$};
		\draw (2,-6) node(f)  [circle, draw=blue, fill=blue!20] {$6$};
		\draw (4,-6) node(g)  [circle, draw=blue, fill=blue!20] {$7$};
		\draw (4,-9) node(h)  [circle, draw=blue, fill=blue!20] {$8$};
		\draw (6,-9) node(i)  [circle, draw=blue, fill=blue!20] {$9$};
		\draw[red] (a.south east) to [short, i>=$e_1$, color=red] (e.north);
		\draw[red] (b.east) to [short, i>=$e_2$, color=red] (e.west);
		\draw[red] (c.south east) to [short, i>=$e_3$, color=red] (f.west);
		\draw[red] (d.south east) to [short, i>=$e_4$, color=red] (i.north);
		\draw[red] (f.east) to [short, i>=$e_6$, color=red] (g.west);
		\draw[red] (f.south east) to [short, i>=$e_7$, color=red] (h.north west);
		\draw[red] (f.north) to [short, i<=$e_5$, color=red] (e.south);
		\draw[red] (g.south) to [short, i<=$e_8$, color=red] (h.north);
		\draw[red] (g.south east) to [short, i>=$e_{9}$, color=red] (i.north west);
	\end{circuitikz}
	\caption{Topology of $\mathcal{G}_1$.}
	\label{fig:G_1}
\end{figure}
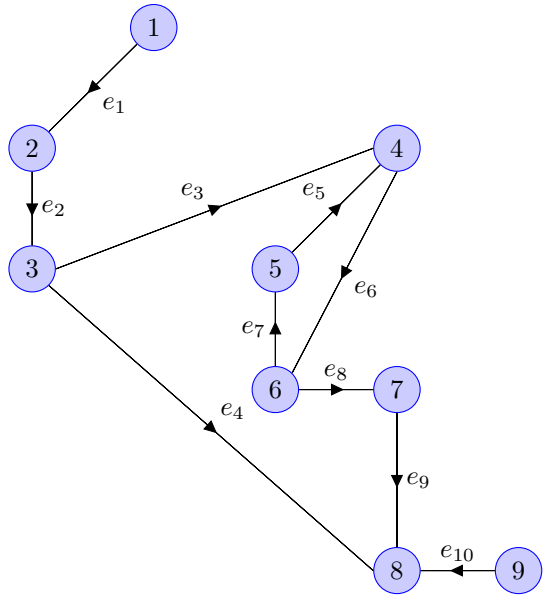
\begin{figure}[H]
	\centering
	\ctikzset{bipoles/length=1.2cm}
	\tikzstyle{every node}=[font=\normalsize]
	
	\begin{circuitikz}[american currents, scale=0.8]
		\draw (0,0) node(a) [circle, draw=blue, fill=blue!20] {$1$};
		\draw (-2,-2) node(b)  [circle, draw=blue, fill=blue!20] {$2$};
		\draw (-2,-4) node(c)  [circle, draw=blue, fill=blue!20] {$3$};
		\draw (4,-2) node(d)  [circle, draw=blue, fill=blue!20] {$4$};
		\draw (2,-4) node(e)  [circle, draw=blue, fill=blue!20] {$5$};
		\draw (2,-6) node(f)  [circle, draw=blue, fill=blue!20] {$6$};
		\draw (4,-6) node(g)  [circle, draw=blue, fill=blue!20] {$7$};
		\draw (4,-9) node(h)  [circle, draw=blue, fill=blue!20] {$8$};
		\draw (6,-9) node(i)  [circle, draw=blue, fill=blue!20] {$9$};
		
		\draw[black] (a.south west) to [short, i>=$e_1$, color=black] (b.north east);
		\draw[black] (b.south) to [short, i>=$e_2$, color=black] (c.north);
		\draw[black] (c.east) to [short, i>=$e_3$, color=black] (d.west);
		\draw[black] (c.south east) to [short, i>=$e_4$, color=black] (h.west);
		\draw[black] (d.south west) to [short, i_<=$e_5$, color=black] (e.north east);
		\draw[black] (d.south) to [short, i>=$e_6$, color=black] (f.north east);
		\draw[black] (f.east) to [short, i>=$e_8$, color=black] (g.west);
		\draw[black] (f.north) to [short, i>=$e_7$, color=black] (e.south);
		\draw[black] (g.south) to [short, i>=$e_9$, color=black] (h.north);
		\draw[black] (h.east) to [short, i<=$e_{10}$, color=black] (i.west);
	\end{circuitikz}
	\caption{Topology of $\mathcal{G}_2$.}
	\label{fig:G_2}
\end{figure}
According to the assigned edge directions, the incidence matrices associated with $\mathcal{G}_1$ and $\mathcal{G}_2$ have the form
\begin{equation*}
B_1 = \matr{
	-1  &   0  &   0   &  0   &  0  &   0   &  0  &  0   &  0\\
	0  &  -1   &  0    & 0   &  0   &  0    & 0   &  0   &  0\\
	0  &   0   &-1  &   0  &   0  &   0  &   0  &   0  &   0\\
	0  &   0   &  0  &  -1   &  0  &   0  &   0   &  0   &  0\\
	1  &   1   &  0  &   0   & -1  &   0  &   0   &  0   &  0\\
	0  &   0   &  1  &   0   &  1  &  -1  &  -1   &  0   &  0\\
	0  &   0   &  0  &   0   &  0  &   1   &  0   &  1   & -1\\
	0  &   0   &  0  &   0  &   0  &   0  &   1  &  -1  &   0\\
	0  &   0   &  0   &  1   &  0  &  0  &   0   &  0  &   1}
\end{equation*}
and
\begin{equation*}
B_2 = \matr{
	-1&	0&	0&	0&	0&	0&	0&	0&	0&	0\\
	1&	-1&	0&	0&	0&	0&	0&	0&	0&	0\\
	0&	1&	-1&	-1&	0&	0&	0&	0&	0&	0\\
	0&	0&	1&	0&	1&	-1&	0&	0&	0&	0\\
	0&	0&	0&	0&	-1&	0&	1&	0&	0&	0\\
	0&	0&	0&	0&	0&	1&	-1&	-1&	0&	0\\
	0&	0&	0&	0&	0&	0&	0&	1&	-1&	0\\
	0&	0&	0&	1&	0&	0&	0&	0&	1&	1\\
	0&	0&	0&	0&	0&	0&	0&	0&	0&	-1
},
\end{equation*}
respectively. As regards the (positive) weights of the edges of $\mathcal{G}_1$ and $\mathcal{G}_2$, they are collected in the following diagonal matrices
\[W_1 = \mathrm{diag}[0.8842, 0.8676, 0.9167, 0.8456, 0.2113, 0.0038, 0.4139, 0.1918, 0.9815]\]
and
\small
\[W_2 = \mathrm{diag}[0.6074, 0.9785, 0.8275, 0.3907, 0.4405, 0.2719, 0.1663, 0.8310, 0.3885, 0.8292].\]
\normalsize
At this point, we have all the ingredients for computing the Laplacians of $\mathcal{G}_1$ and $\mathcal{G}_2$ (called $\mathbb{L}$ and $\mathbb{M}$, respectively) as
\begin{equation*}
\begin{aligned}
\mathbb{L} &= B_1 W_1 B_1^T = \\
&=\matr{    0.8842    &     0     &    0     &    0  & -0.8842      &   0    &  0   &      0      &   0\\
	0  &  0.8676   &      0     &    0  & -0.8676    &     0   &      0    &     0    &     0\\
	0     &    0  &  0.9167  &       0     &    0  & -0.9167    &     0    &     0   &      0\\
	0     &    0  &       0  & 0.8456    &     0   &      0     &    0     &    0  & -0.8456\\
	-0.8842  & -0.8676  &       0   &     0  &  1.9631  & -0.2113  &       0    &     0    &     0\\
	0    &     0  & -0.9167    &     0 &  -0.2113  &  1.5458  & -0.0038 &  -0.4139  &       0\\
	0    &     0   &      0     &    0   &      0 &  -0.0038  &  1.1771 &  -0.1918 &  -0.9815\\
	0    &     0   &      0     &    0   &      0 &  -0.4139  & -0.1918 &   0.6057 &        0\\
	0    &    0    &     0   &-0.8456    &     0   &      0  & -0.9815  &       0  &  1.8271 }
\end{aligned}
\normalsize
\end{equation*}
and
\begin{equation*}
\begin{aligned}
\mathbb{M} &= B_2 W_2 B_2^T = \\
& =\matr{  
	0.6074 &  -0.6074    &     0     &    0    &     0     &    0   &      0      &   0  &       0\\
	-0.6074  &  1.5859  & -0.9785      &   0   &      0    &     0  &       0      &   0    &     0\\
	0  & -0.9785  &  2.1967 &  -0.8275  &       0 &       0    &     0 &  -0.3907    &     0\\
	0   &      0  & -0.8275 &   1.5399  & -0.4405  & -0.2719  &       0    &     0   &      0\\
	0   &      0  &       0 &  -0.4405  &  0.6068 & -0.1663   &      0     &    0    &     0\\
	0   &      0  &       0 &  -0.2719  & -0.1663 &   1.2691  & -0.8310   &      0   &      0\\
	0   &      0  &       0 &        0  &       0 &  -0.8310  &  1.2195  & -0.3885   &      0\\
	0   &      0  & -0.3907 &        0  &       0 &        0  & -0.3885 &   1.6085 &  -0.8292\\
	0   &      0  &      0  &       0   &      0  &       0   &      0  & -0.8292  &  0.8292}.
\end{aligned}
\normalsize
\end{equation*}
Next, we pick the positive definite diagonal matrix 
\[D = \mathrm{diag}[0.5977, 0.4297, 0.4937, 0.0058, 0.4643, 0.0005, 0.6299, 0.8209, 0.3597]\]
and pre- and post- multiply it by $\mathbb{L}$ and $\mathbb{M}$, respectively, thus obtaining
\begin{equation*}
\begin{aligned}
\mathbb{Q} &=\mathbb{L}D\mathbb{M} = \\
&=\matr{  0.3210 &  -0.3210    &     0  & 0.1808 &  -0.2491  &  0.0683    &     0   &      0   &      0\\
	-0.2264  &  0.5912 &  -0.3647  &  0.1774 &  -0.2444  &  0.0670   &      0   &      0  &       0\\
	0  & -0.4428  &  0.9941 &  -0.3744 &   0.0001 &  -0.0006  &  0.0004  & -0.1768     &    0\\
	0     &    0  & -0.0040  &  0.0075  & -0.0021 &  -0.0013    &     0  &  0.2522 &  -0.2522\\
	-0.0946  & -0.2701 &   0.3647  & -0.4015   & 0.5531  & -0.1517  &  0.0001    &     0     &    0\\
	0  &  0.4428  & -0.8614 &   0.4175 &  -0.0597 &   0.0193   & 0.1284 &  -0.3688  &  0.2817\\
	0   &      0   & 0.0615 &   0.0000 &   0.0000 &  -0.6161   & 0.9653  & -0.2485  & -0.1622\\
	0   &      0  & -0.1943 &   0.0001  &  0.0000 &   0.1001 &  -0.3403  &  0.8467  & -0.4123\\
	0   &      0  &  0.0040 &  -0.0075 &   0.0021 &   0.5150 &  -0.7539  & -0.3048  &  0.5450}.
\end{aligned}
\normalsize
\end{equation*}
Now, if we compute the eigenvalues of $\mathbb{Q}$, we get 
\[
\begin{aligned}
\text{eig}(\mathbb{Q}) =\{&1.3891 + 0.1564\mathrm{i},1.3891 - 0.1564\mathrm{i},0.9210 + 0.0000\mathrm{i},0.5879 + 0.0000\mathrm{i},\\
&0.4509 + 0.0000\mathrm{i}, 0.1057 + 0.0000\mathrm{i}, \mbf{-0.0002 + 0.0039\mathrm{i}},\\
&\mathbf{-0.0002 - 0.0039\mathrm{i}}, 0.0000 + 0.0000\mathrm{i}\}
\end{aligned}\]
that proves the desired result. This example led us to introduce Assumptions \ref{ass:D_identity} and \ref{ass:same_topology}, which define conditions under which the eigenvalues of $\mathbb{Q}$ are always real and nonnegative.
	\newpage
     \section{Electrical parameters }
 \label{sec:AppElectrPar}
 In this appendix, we provide the electrical parameters of the simulation scenario described in Section \ref{sec:simulations}.
 
 
 \begin{table}[h]
 	\centering
 	\begin{tabular}{*{4}{c}}
 		\toprule
 		\multicolumn{4}{c}{\textbf{Converter parameters}} \\
 		\midrule
 		DGU & $R_t$ $(\Omega)$ & $L_{t}$ (mH) & $C_t$ (mF)\\
 		\midrule
 		$\subss{\hat{\Sigma}}{1}$ & 0.2  & 1.8 & 2.2 \\
 		$\subss{\hat{\Sigma}}{2}$& 0.3 & 2 & 1.9 \\
 		$\subss{\hat{\Sigma}}{3}$& 0.1 & 2.2 & 1.7 \\
 		$\subss{\hat{\Sigma}}{4}$& 0.5 & 3 & 2.5 \\
 		$\subss{\hat{\Sigma}}{5}$& 0.4 & 1.2 & 2 \\
 		$\subss{\hat{\Sigma}}{6}$& 0.6 & 2.5 & 3 \\
 		$\subss{\hat{\Sigma}}{7}$& 0.3 & 2 & 2.1 \\
 		\midrule
 		\multicolumn{4}{c}{\textbf{Power line parameters}}\\
 		\toprule
 		\multicolumn{2}{c}{Connected DGUs $(i,j)$} & Resistance $R_{ij}$ $(\Omega)$ & Inductance
 		$L_{ij} (\mu$H$)$ \\
 		\midrule
 		\multicolumn{2}{c}{$(1,2)$} & 0.05 & 2.1 \\
 		\multicolumn{2}{c}{$(1,3)$} & 0.07 & 1.8 \\
 		\multicolumn{2}{c}{$(3,4)$} & 0.06 & 1 \\
 		\multicolumn{2}{c}{$(2,4)$} & 0.04 & 2.3 \\
 		\multicolumn{2}{c}{$(4,5)$} & 0.08 & 1.8 \\
 		\multicolumn{2}{c}{$(1,6)$} & 0.1 & 2.5 \\
 		\multicolumn{2}{c}{$(5,6)$} & 0.08 & 3 \\
 		\multicolumn{2}{c}{$(4,7)$} & 0.09 & 2.3 \\
 		\multicolumn{2}{c}{$(7,5)$} & 0.05 & 2.4 \\
 		\bottomrule
 	\end{tabular}
 	\caption{Electrical parameters.}	
 	\label{tbl:Ch5_el_parameters}
 \end{table}
 

     \clearpage

   \bibliographystyle{IEEEtran}
     \bibliography{PnP_consensus}

\end{document}